\documentclass[tikz,12pt]{article}
\usepackage{mheckII}

\usepackage[T1]{fontenc}

\usepackage{tikz-feynman}

\usepackage{manfnt}
\usepackage{longtable}
\usepackage{amsxtra}

\usepackage{fancybox}

\usepackage{multirow}

\usepackage{blkarray}

\usepackage{tikz} 
\usetikzlibrary{matrix}

\theoremstyle{theorem}
\newtheorem{thm}{Theorem}

\newtheorem{corl}{Corollary}
\newtheorem{pro}{Proposition}

\newtheorem{lem}{Lemma}

\newtheorem{correspondence}{Correspondence}

\newtheorem*{exce}{Exception}

\newtheorem{claim}{Claim}

\newtheorem*{proposal}{Proposal}

\theoremstyle{definition}

\newtheorem{rem}{Remark}

\newtheorem*{war}{Warning}

\def\Dsl{\,\raise.15ex\hbox{/}\mkern-13.5mu D}
\def\dsl{\,\raise.25ex\hbox{/}\mkern-10.5mu \partial}

\def\bze{\boldsymbol{\zeta}}

\title{Fuchsian ODEs as Seiberg dualities
}

\authors{Sergio Cecotti\footnote{e-mail: {\tt cecotti@sissa.it}, {\tt cecotti@bimsa.cn}}\vskip 9pt

\centerline{Yanqi Lake Beijing Institute of Mathematical Sciences and Applications (BIMSA)\footnote{On leave from SISSA,
via Bonomea 265, Trieste, Italy}}
\centerline{Yanqi Island, Huairou District, Beijing 101408, China}
}

\abstract{The classical theory of Fuchsian differential equations is largely equivalent to the theory of Seiberg dualities for quiver SUSY gauge theories. In particular: \emph{all} known integral representations of solutions, and
their connection formulae, are immediate consequences of (analytically continued) Seiberg duality in view of the dictionary between linear ODEs and gauge theories with 4 supersymmetries.

The purpose of this divertissement is to explain ``physically'' this remarkable relation in the spirit of Physical Mathematics. The connection goes through a ``mirror-theoretic'' identification of 
 irreducible logarithmic connections on $\mathbb{P}^1$ with would-be BPS dyons of
4d $\cn=2$ $SU(2)$ SYM coupled to a certain Argyres-Douglas ``matter''. When the underlying bundle
is trivial, i.e.\! the log-connection is a Fuchs system, the world-line theory of the dyon simplifies and the action of Seiberg duality on the Fuchsian ODEs becomes
quite explicit. The duality action is best described in terms of Representation Theory of Kac-Moody Lie algebras
(and their affinizations). 
}

\begin{document}
\maketitle

\tableofcontents


\section{Introduction and Overview}

A \emph{Fuchsian ODE}\footnote{\ For a recent textbook survey see \cite{book}.} is a linear differential equation of order $n$ on the Riemann sphere $\mathbb{P}^1$
\be\label{ODE1}
\sum_{k=0}^n a_k(z)\,\frac{d^k y}{dz^k}=0
\ee
with polynomial coefficients $a_k(z)$
and \emph{regular singularities}  in a finite set of points $S\equiv\{z_1,\cdots\!, z_s\}\subset\mathbb{P}^1$ \cite{deligne}.
It can 
be written\footnote{\ Strictly speaking this can be done if some mild condition is satisfied, e.g.\! if the monodromy representation is irreducible \cite{beauville}. We can always put the ODE in the form \eqref{ODE2} by adding an extra point to the set $S$ with trivial local monodromy (`apparent singularity') \cite{beauville}.} as 
a \emph{Fuchsian system}  i.e.\! a system of $n$ first-order equations of the form
\be\label{ODE2}
\nabla\mspace{1mu} Y(z)=0\quad\text{where}\quad \nabla\equiv \frac{d}{dz}+\sum_{i=1}^s\frac{A_i}{z-z_i},
\ee 
where $Y(z)$ is a $n$ component vector of unknown functions and
the $A_i$'s are constant $n\times n$ matrices which satisfy a certain \emph{non-resonant} condition\footnote{\ Let $w_{i,\ell}$ be the eigenvalues of $A_i$. The non-resonant condition says that if $w_{i,\ell^\prime}=w_{i\ell}\bmod1$ then $w_{i,\ell^\prime}=w_{i,\ell}$.
In the language of \cite{deligne}: the eigenvalues of $A_i$ should belong to a \emph{transversal} $T_i$ of $\C/\Z$.}. 

\subparagraph{Logarithmic connections.} More generally, 
we consider pairs $(E,\nabla)$ where $E\to\mathbb{P}^1$ is a holomorphic vector bundle of rank $n$
and 
\be
\nabla\colon E\to E\otimes\Omega^1(\log S)
\ee
 is a \emph{logarithmic connection} \cite{deligne} on $E$ i.e.\! a connection which locally takes the form \eqref{ODE2}. $\nabla$ is automatically flat in one dimension, $\nabla^2=0$, and the general ODE
 \be\label{ODE3}
 \nabla\Psi=0
 \ee
is integrable.  A pair $(E,\nabla)$ is a classical Fuchsian system iff $\Psi$ is a vector of functions $Y(z)$, that is, if $E$ is the \emph{trivial rank-$n$ bundle} $\co^{\, n}$.
 
The Riemann-Hilbert correspondence (RHC) identifies (up to equivalence)
 the general order-$n$ ODE \eqref{ODE3}
with its \emph{monodromy representation} \cite{deligne}
\be
\varrho\colon\pi_1(\mathbb{P}^1\setminus S,\ast)\to GL(n,\C).
\ee
The image $\Gamma\subset GL(n,\C)$ is the \emph{monodromy group} and its Zariski-closure in the complex algebraic group $GL(n,\C)$ the \emph{differential Galois group} of the ODE \cite{diffGalois}. 
$\varrho$ (resp.\! $\Gamma$) is independent of the chosen base point $\ast$ up to conjugacy,
and we identify monodromy representations (resp.\! groups)
up to conjugation in $GL(n,\C)$. 
When the representation $\varrho$ is \emph{reducible} (that is, preserves a non-trivial subspace of $\C^n$)
 the ODE can be reduced to ODEs of lower order \cite{book}.
When convenient we shall restrict to irreducible ODEs
without essential loss. 
\medskip

We shall say (a bit abusively) that an ODE is ``unitary'' iff the residue matrices $A_i$ are Hermitian $A_i=A_i^\dagger$.
We introduce this notion mainly for didactical reasons: the adjective ``unitary'' here refers to a property of the
physical system associated to the ODE. However in the situations of main interest for the applications this will coincide with the usual notion that $\nabla$ is a unitary connection i.e.\! $\Gamma\subset U(n)$.
The general case where the $A_i$ are arbitrary complex matrices
may be obtained  as the ``complexification'' of the ``unitary'' case (modulo important subtleties).
Our strategy is to understand first the simpler ``unitary'' case and then analytically continue to the general situation.

\medskip

In view of the RHC, a general ODE \eqref{ODE3} is specified by two sets of data:
\begin{itemize}
\item[\bf (I)] the marked points $S=\{z_1,\cdots\!,z_s\}\subset \mathbb{P}^1$ together with the $s\equiv |S|$ conjugacy classes $C_i\subset GL(n,\C)$
of the \emph{local monodromies} $\varrho_i\equiv\varrho(\ell_i)$ along small loops $\ell_i$ encircling the marked points $z_i\in S$. A marked point $z_j$ with trivial local monodromy, $\varrho_j=1$, is called an \emph{apparent singularity:} in this case $A_j=h\!\cdot\!\boldsymbol{1}$ with $h\in\Z$. A marked point
$z_i\in S$ with $\varrho_i\neq 1$ is an \emph{essential} singularity;
\item[\bf (II)] a point $a$ in the moduli space $\mathscr{A}$
of dimension $n$ complex representations of  
\be
\pi_1(\mathbb{P}^1\setminus S)\equiv\langle \ell_1,\cdots\!,\ell_s\;|\; \ell_1\ell_2\cdots\ell_s=1\rangle,
\ee
with prescribed conjugacy classes $C_i\subset GL(n,\C)$ for the generators $\varrho_i\equiv\varrho(\ell_i)$, modulo overall conjugacy:
\be
\mathscr{A}\overset{\rm def}{=}\big\{\varrho_i\in GL(n,\C),\ i=1,\cdots\!,s \;\big|\; \varrho_1\varrho_2\cdots\varrho_s=1\ \text{and }\varrho_i\in C_i\big\}\big/GL(n,\C).
\ee
\end{itemize}

When $\mathscr{A}=\varnothing$ there is no $(E,\nabla)$ pair with the given $C_i$'s.
In facts there is a \emph{topological obstruction} to the existence of a pair $(E,\nabla)$:
the pair exists (i.e.\! $\mathscr{A}\neq\varnothing$) ionly if\,\footnote{\ Eq.\eqref{i987weqq0} is obvious when the connection is unitary ($\Gamma\subset U(n)$). In this case the first Chern class $c_1(E)$ of $E$ is represented by 
$(2\pi i)^{-1}$ times the curvature $(1,1)$-current $\overline{\partial}\mspace{1.5mu}\mathrm{tr}(\Phi^{-1}\partial\Phi)$ where $\Phi$ is a fundamental solution of the ODE. Using the
Poincar\'e-Lelong formula we conclude that $c_1(E)$ is represented by the current $-\sum_i \mathrm{tr}\,A_i\, \delta^{(2)}(z-z_i)$;
integrating over $\mathbb{P}^1$ we get eq.\eqref{i987weqq0}. Exponentiating one gets the identity $1=\prod_i\det(\varrho_i)\equiv\exp(-2\pi i\sum_i\mathrm{tr}\,A_i)$ which follows from the
relation $\ell_1\cdots\ell_s=1$ in $\pi_1(\mathbb{P}^1\setminus S)$. The formula remains true in the general case \cite{CB3}.}
\be\label{i987weqq0}
\deg E+\sum_{i=1}^s \mathrm{tr}\,A_i=0
\ee
(this condition is necessary;\footnote{\ I thank Alexander Soibelman for comments on this issue.}
for the sufficient condition see \S.\,7 of \cite{CB3}).

An irriducible ODE $(E,\nabla)$ has \emph{rigid monodromy} iff $\mathscr{A}$ is a point. When this is the case, we say that the ODE is \emph{rigid.}
The moduli space $\mathscr{A}$ admits a stratification
\be
\mathscr{A}=\coprod\nolimits_{\bze} \ca(\bze),\qquad \ca(\bze)\overset{\rm def}{=}\big\{\Gamma\in\mathscr{A}\;\big|\; \bze(\Gamma)=\bze\big\}
\ee 
where the complex Lie group $\bze(\Gamma)\supset GL(1,\C)$ is the centralized of $\Gamma$ in $GL(n,\C)$.
We write $\ca$ for the \emph{generic} stratum $\ca(\bze_\text{min})$, i.e.\! the open dense domain
with the smallest possible centralizer $\bze_\text{min}$. 
One has \cite{book}:
\be\label{887zzaq123}
\dim_\C \ca= 1+(s-2)n^2-\sum_{i=1}^s \dim_\C \bze(A_i)+\dim_\C \bze_\text{min}\in 2\,\mathbb{N},
\ee 
where $\bze(A_i)\subset GL(n,\C)$ is the subgroup of matrices commuting with the residue matrix $A_i$. The ODE \eqref{ODE3} is \emph{irreducible} along $\ca$ iff $\bze_\text{min}=\C^\times$.
From eq.\eqref{887zzaq123} we see
that in the non-trivial case ($s\geq3$ and $n\geq2$) the only ODE which is rigid for \emph{generic} $C_i$'s is the order-2 hypergeometric.
In the physical applications we are mainly interested in rigidity for \emph{non-generic} $C_i$'s (see \S.\,\ref{s:exam} below).
According to the $19^\text{th}$ century tradition, in the Fuchs case the generic stratum $\ca$ is called the space of
\emph{accessory parameters}. We extend this terminology to the more general ODE \eqref{ODE3}. 

\subparagraph{ODEs of geometric/physical origin.}
The ODEs \eqref{ODE3} which arise from geometry\footnote{\ The typical examples are the Picard-Fuchs equations 
satisfied by the periods of one-parameter families of projective varieties.} (and physics) have special properties (see e.g.\! \cite{peters}):
\begin{itemize}
\item[(A)] the local monodromies $\varrho_i$ are \emph{quasi-unipotent}: $(\varrho_i^{m_i}-1)^{k_i+1}=0$ for some $m_i,k_i\in\mathbb{N}$;
\item[(B)] the monodromy group $\Gamma$
is \emph{semisimple} and $\Gamma\subset SL(n,\C)$. In facts in the geometric/physical set-up we have the stronger condition
$\Gamma\subseteq M(\Z)$ where $M(\R)\subsetneqq SL(n,\R)$ is a semisimple, real, Lie subgroup of Hodge type\footnote{\ That is, a semisimple real Lie group without compact simple factors having a compact maximal torus. More precisely: in the geometric set-up $\Gamma$ is defined over $\mathbb{Q}$ and
$M$ is the \emph{$\mathbb{Q}$-algebraic monodromy} (a.k.a.\! the \emph{Mumford-Tate group} \cite{simpson,griffiths1,griffiths2,periods}) given by the Zariski closure of $\Gamma$ over $\mathbb{Q}$. $M(\R)$ is the real Lie group of the $\R$-valued points of $M$.} \cite{simpson,griffiths1,griffiths2,periods}. 
Without loss $\varrho$ may be assumed to be \emph{irreducible}.
Recall that $\varrho$ is irreducible iff $\dim_\C \boldsymbol{\zeta}(\Gamma)=1$.
\end{itemize}

For logarithmic connections on a curve the Deligne-Simpson conjecture \cite{simpson}
is a theorem \cite{katz}: an irreducible ODE \eqref{ODE3} over $\mathbb{P}^1$ which is rigid and has quasi-unipotent
local monodromies is \emph{motivic}, that is, its solutions are the periods
$\int_\gamma \lambda$ of an \emph{algebraic} differential form $\lambda$. 
Equivalently: the rigid monodromy groups are \emph{integral} i.e.\!
$\Gamma\subset GL(n,\mathfrak{o}_\mathbb{F})$ where $\mathfrak{o}_\mathbb{F}$ is the ring of integers in some number field $\mathbb{F}$.
This fundamental math result is (essentially)
a consequence of Seiberg duality: see  \S.\,\ref{s:solution} below for the explicit $\lambda$ in the
Fuchsian case. 
\medskip

In the geometric/physical set-up the eigenvalues $\xi_{i,\ell}$
of the local monodromies $\varrho_i$ belong to the unit circle. For elegance of the physical interpretation 
 we always assume $|\xi_{i,\ell}|=1$ 
although mathematically the statements hold independently of this assumption.
\medskip

 In the geometric/physical applications $\varrho$ takes values in
$SL(n,\C)$ so $\deg E=0$.
In view of the applications to physical problems, in this divertissement we mainly focus on
the important (and simpler) case where $\deg E=0$. However in section 2 we shall sketch the general story from our ``physical'' standpoint.

\subparagraph{Fuchsian systems.} The Fuchsian systems are 
 the simpler situations where $E$ is the trivial rank $n$ bundle $\co^{\,n}$. Let $\Phi(z)$
 be the \emph{fundamental solution} of \eqref{ODE2}, i.e.\! a $n\times n$ matrix such that $\nabla\Phi(z)=0$
 normalized by the condition $\Phi|_\ast=\boldsymbol{1}_n$ at a base point $\ast\in\mathbb{P}^1\setminus S$.
The Maurier-Cartan 1-form $(d\Phi)\Phi^{-1}$ is then meromorphic on
$\mathbb{P}^1$, so its total residue vanishes
\be\label{Afucchs}
A_1+A_2+\cdots +A_s=0.
\ee
The eigenvalues of $A_i$ are called \emph{(local) exponents} of the ODE. Since we are assuming
$|\xi_{i,\ell}|=1$, for us the exponents are real numbers.

When $E$ is trivial, the accessory parameter space $\ca$
is a holomorphic symplectic manifold with $(2,0)$ form $\Omega$.
If the exponents are real,  we have an \emph{anti}-holomorphic involution
\be
\iota\colon \ca\to\overline{\ca}
\ee
 which acts as $\varrho\mapsto (\varrho^\dagger)^{-1}$ on the monodromy representation and as complex
 conjugation on the 2-form $\Omega$. 

\subparagraph{``Unitary'' Fuchsian ODEs: the SUSY model and its Higgs branch.} It is convenient to consider first ``unitary'' Fuchsian systems with Hermitian residue matrices
since in this set-up the physical picture is pretty standard. The general case will then be obtained by ``complexification'' (in the sense of Cartan's complexification of real-analytic manifolds \cite{cartan}).  A necessary condition for the existence of a ``unitary'' ODE
is that all classes $C_i$ contain unitary elements
\be\label{unittt}
C_i\cap U(n)\neq\varnothing\quad \text{for all $i$.}
\ee
Our first claim is:

\begin{correspondence}\label{corr1} In the Fuchsian set-up, $E\simeq\co^{\,n}$, there is a natural correspondence
\be
(\co^{\,n},\nabla)\longleftrightarrow (\mathscr{Q},v)
\ee
where $(\co^{\,n},\nabla)$ is a Fuchsian system with monodromy $\Gamma\subset U(n)$, $\mathscr{Q}$ is a 4-SUSY quiver quantum mechanical system  {\rm (SQM)} defined\,\footnote{\ While $\mathscr{Q}$ is not unique, the Higgs branches of the several SQMs associated to an ODE
are all isometric, that is, $\mathscr{Q}$ is unique up to a group of \emph{trivial} IR dualities which provide the natural equivalence in the game.} by the ODE datum {\bf (I)}, and $v$ is a point in its space  $\cm(H)$ of 
SUSY-preserving (classical) vacua where the gauge group $G$ is broken to
\be
H\equiv \bze(\Gamma)\cap U(n)\supseteqq U(1).
\ee 
The ODE is \emph{irreducible} iff $H=U(1)$. 
\end{correspondence}

Let $H_\text{min}\equiv \bze_\text{min}\cap U(n)$. We write $\cm$ for
$\cm(H_\text{min})$ and call it the \emph{Higgs branch} of $\mathscr{Q}$.
We say that $\cm$ is \emph{fully Higgsed} iff $H_\text{min}=U(1)$. 
An irreducible ``unitary'' ODE then corresponds to a point $v$ in a fully Higgsed branch.
Since we are free to assume the ODE to be irreducible,
we lose nothing essential if we consider only fully Higgsed branches, forgetting all other branches of the SUSY vacuum space. 
\medskip

 The Lagrangian $L$ of $\mathscr{Q}$ is formally the dimensional reduction to 1d of a 4d $\cn=1$ quiver gauge theory.
The parent 4d QFT may be anomalous and/or non UV-complete, but the 1d SQM has no pathology. The reader may feel more natural to consider the 3d $\cn=2$ version of $\mathscr{Q}$
 which is also a nice QFT. The 3d perspective has the advantage that $\cm$ parametrizes actual \emph{quantum} vacua not classical ones as in the SQM situation. Our emphasis on the 1d viewpoint mainly reflects the author's own prejudices.
 \medskip
 
  In order for \textbf{Correspondence \ref{corr1}} to make sense, SUSY should \emph{not} be spontaneously broken in $\mathscr{Q}$. 
This puts a restriction on the Fayet-Iliopoulos (FI) couplings in the 1d Lagrangian $L$. 
Going back for a moment to the general logarithmic connection \eqref{ODE3},
where a similar (but more complicated) correspondence is expected to hold (cf.\! \S.\,2),
one sees that the condition on the FI couplings for unbroken SUSY
 is \emph{precisely equivalent} to setting to zero
 the topological obstruction \eqref{i987weqq0} to the existence of the pair $(E,\nabla)$. 
 \medskip

The space $\overline{\cm}$ of \emph{all} SUSY vacua of $\mathscr{Q}$ is a \emph{projective complex variety}, and the Higgs branch $\cm\subset \overline{\cm}$ carries a (possibly non-complete) natural K\"ahler metric. By \textbf{Correspondence \ref{corr1}} $\cm$ is
the locus in the accessory space $\ca$ which parametrize ``unitary'' Fuchsian systems with the given local monodromy classes $C_i$ satisfying \eqref{unittt}. 
In facts, when \eqref{unittt} holds,
the Higgs branch $\cm$ (seen as a real-analytic manifold) is the \emph{fixed locus} of $\ca$ under
the complex conjugation $\iota$
and  
\be
\dim_\C\ca=2\dim_\C\cm.
\ee $\Omega|_{\cm}$ is then a real symplectic form to be identified
with the K\"ahler $(1,1)$-form of $\cm$.
In view of these facts, when \eqref{unittt} holds
the accessory parameter space $\ca$ is the natural ``complexification''
of the Higgs branch $\cm$ of the associated SQM $\mathscr{Q}$.

%

\subparagraph{Infra-red dualities.} By its very definition, 
an infra-red (IR) duality $\mathscr{Q}\leadsto \mathscr{Q}^{\mspace{1mu}\prime}$
between 
two quiver SQMs yields \emph{K\"ahler isometries} $\phi_H\colon\cm(H)\to\cm(H)^\prime$ between the several branches of its SUSY vacuum space $\overline{\cm}$.
We focus on the isometry $\phi\colon \cm\to\cm^\prime$ between fully Higgsed branches.
The IR dualities for quiver SQM/3d\,QFT 
are compositions of basic \emph{Seiberg dualities} \cite{SEI1,SEI2,SEI3,SEI4}.
\medskip
 
Let $\mathscr{Q}$ be the quiver SQM associated to a ``unitary'' Fuchsian ODE
for a certain datum $\textbf{(I)}$. We say that the Seiberg duality
$\mathscr{Q}\leadsto\mathscr{Q}^{\mspace{1mu}\prime}$ is
 \emph{admissible} iff the dual SQM $\mathscr{Q}^{\mspace{1mu}\prime}$ also corresponds to
 a ``unitary'' Fuchsian ODE for some datum $\textbf{(I)}^\prime$.
 The groupoid of admissible Seiberg dualities acts 
 on the space of ``unitary'' Fuchsian ODEs by
 \be\label{uuuuuy76}
 (\mathscr{Q},v)\leadsto (\mathscr{Q}^{\mspace{1mu}\prime},\phi(v)).
 \ee
 
 Non-trivial IR dualities change the order $n$ and the number $s$ of \emph{essential} regular singularities,\footnote{\ Seiberg dualities act trivially on the set $S\subset\mathbb{P}^1$ of marked points, but they may transform an essential singularity into an apparent one which we happily forget (see \S.\,\ref{s:internal} below).}
thus potentially leading to a simpler differential equation. 
All structural properties of a Fuchsian ODE such as, irreducibility, rigidity, isomonodromic deformations,
etc.,
are preserved by the IR duality.

What is more, Seiberg duality applies functorially
to the space of 
\emph{solutions} of the ODEs: if we know the fundamental solution $\Phi(z)$ of one Fuchsian equation,
we can explicitly write the solution to any other ODE which can be obtained from it
by a chain of admissible Seiberg dualities, using an algorithm originally due to N. Katz called \emph{middle convolution}
\cite{katz} (see also \cite{convolution,convolution1,convolution2}; for a textbook treatment see \cite{book}). In particular, a ``unitary'' Fuchs ODE has rigid monodromy if and only if there
is a chain of admissible Seiberg dualities which maps it to the first order ODE whose associated SQM/3d\,QFT is \emph{free} (hence trivially solvable).

When the monodromy is rigid, this leads to an explicit integral representation of the solutions\footnote{\ See e.g.\! \textbf{Theorem 7.23} of \cite{book}; for more details \cite{AA1}. The story is reviewed in \S.\,\ref{s:star} below.}, and we can write their connection formulae explicitly \cite{AA2,AA2b}.
This consequence of Seiberg duality was already observed by Riemann (using a slightly different language) in his 1851 lecture notes on the hypergeometric equation, and played a central role in the mathematics of the following $1.7$ centuries.

\subparagraph{The general Fuchsian case: ``complexification''.}
When  the eigenvalues $\xi_{i,\ell}$
are \emph{generic} points on the unit circle (satisfying the topological constraint $\prod_{i,\ell}\xi_{i,\ell}=1$)  the condition \eqref{unittt} holds automatically and
there is a nice associated SQM $\mathscr{Q}$. In this situation the extension from $\Gamma\subset U(n)$ to a general monodromy $\Gamma\subset GL(n,\C)$
just requires to replace in  \textbf{Correspondence \ref{corr1}} the Higgs branch $\cm$ with its ``complexification'' $\ca$:
\be\label{uuyyqart}
(\co^{\,n},\nabla)\longleftrightarrow (\mathscr{Q},a),\qquad a\in\ca.
\ee

The ``complexified'' Higgs branch $\ca$ is best seen as the \emph{ordinary} Higgs branch
of a ``complexified'' quiver SQM $\mathscr{Q}^{\mspace{2mu}\C}$ obtained from $\mathscr{Q}$
by the following procedure:
\begin{itemize}
\item[(a)] declare the anti-chiral superfields $X^*_v$ to be independent complex superfields
instead of being the Hermitian conjugates $\overline{X}_v$ of the chiral superfields $X_v$;
\item[(b)] replace the compact gauge group $G$ by its complexification $G^{\mspace{2mu}\C}$, equivalently
make all vector superfields $V_\sigma$ complex.
\end{itemize}
One may consider more general models where also the FI couplings are complex: this will allow for complex exponents in the Fuchsian equation. 
We shall not
use this freedom. 

$\mathscr{Q}^{\mspace{1mu}\C}$ has the anti-holomorphic involution $\iota$
\be
X^*_v \longleftrightarrow \overline{X}_v,\qquad V_\sigma\longleftrightarrow \overline{V}_\sigma.
\ee
The Higgs branch $\ca$ of $\mathscr{Q}^{\mspace{2mu}\C}$ is defined as its space of SUSY preserving classical vacua where the gauge group $G^{\mspace{1mu}\C}$ is broken to
$\bze_\text{min}$.  Again, we may assume $\bze_\text{min}=\C^\times$ with no loss.

The complexified model $\mathscr{Q}^{\mspace{2mu}\C}$ is clearly \emph{non-unitary}.
Non-unitary SQMs make sense in full generality, even when condition \eqref{unittt}
is \emph{not} satisfied. In this case the ``unitary'' locus $\cm=\varnothing$,
but $\ca$ is still a good holomorphic symplectic manifold of dimension \eqref{887zzaq123}.
The only difference with the case \eqref{unittt} is that now $\iota$ has no fixed point in $\ca$
(but only in $\mathscr{A}\equiv\overline{\ca}$).
Then

\begin{correspondence} To a ``unitary'' ODE there corresponds a unitary SQM.
To a ``non-unitary'' OPE there corresponds a \emph{non-unitary} SQM which is the
``complexification'' of a standard 4-SUSY quiver SQM.
The Higgs branch of the non-unitary SQM is the accessory parameter space $\ca$.
Along $\ca$ the complexified gauge group is broken down to $\bze_\text{\rm min}$.
The monodromy is irreducible iff $\bze_\text{\rm min}=\C^\times$.
\end{correspondence} 

The analytically continued Seiberg dualities act on the space of non-unitary quiver SQM
by biholomorphic maps $\ca\to \ca^\prime$. Again, we may use this action to simplify the
ODE. The unbroken gauge group $\bze_\text{\rm min}$
is an invariant of the duality,  so rigid ODEs are mapped into rigid ones.
There is a chain of dualities connecting the given ODE to the $1^\text{st}$ order one (which is trivially rigid) if and only if the original ODE was rigid. If we know the solution to one ODE,
we can write the solutions to any other one in the same complexified Seiberg orbit, 
write explicitly its connection formulae, etc. 
In the non-rigid case, Seiberg duality identifies the \emph{isomonodromic  deformations} of dual ODEs  as we vary the $s$-tuple of marked points $\{z_1,\cdots\!,z_s\}$
\cite{AA3}.

\subparagraph{Why writing this story?} This \emph{divertissement} belongs to the discipline that
Greg Moore calls ``Physical Mathematics'', that is, the art of deducing/explaining/clarifying purely
mathematical results using insights from physics.  In particular this note does not contain any \emph{essentially new}  mathematical theorem.
Its purpose is to give a conceptual ``physical'' explanation of why the theory of Fuchsian ODEs and the theory of IR dualities in 4-supercharge systems ought to be largely one and the same,
and also show how tricky (known) math theorems are obvious from the physical side
although they look surprising (to say the least) from the original mathematical perspective.

 The ``physical'' correspondence is dictated by homological mirror symmetry and 
 goes through the BPS dyons of a certain class of 4d $\cn=2$ gauge theories
constructed in \cite{classification} and further studied in \cite{think,shep}. 

We stress that our ``physical'' story is merely a \emph{reinterpretation} in a suggestive
 field-theoretic language of the rigorous  mathematical approach to the Deligne-Simpson problem \cite{sim1,sim2,sim3,sim4,sim5,sim6}
due to W. Crawley-Boevey \cite{CB1,CB2,CB3,CB4,CB5,CB6}
which builds on the Representation Theory of deformed pre-projective algebras \cite{PP1,PP2,PP3,PP4}.
For a math textbook treatment see \cite{book}.

\subparagraph{Organization.}This note is organized as follows. In section 2 we sketch the very natural
correspondence between ODEs and $\cn=2$ BPS dyons using the BPS quiver approach \cite{bpsqui,4d/2d}  together with
results from \cite{classification,think,galois} and their relation to the Representation Theory of squid algebras \cite{CB3,squid}
and affine extensions of Kac-Moody Lie algebras \cite{alg2,alg3}. 
In section 3 we consider the associated 1d SQM following \cite{classification,bpsqui,think}, describe its complexification and its relation to the modules of the associated deformed pre-projective algebra. Then
we define the admissible Seiberg dualities and construct explicitly the isometry $\phi$ in eq.\eqref{uuuuuy76}.
In section 4 we briefly review how Seiberg duality acts on the \emph{solutions} of the Fuchsian ODEs and their connection coefficients. Section 5 contains a list of examples. We defer a minor technicality to the appendix.

\medskip

\section{Logarithmic connections vs.\! BPS dyons}\label{s:second}

In this section $E$ is any rank-$n$ holomorphic bundle over $\mathbb{P}^1$
endowed with a non-resonant logarithmic connection $\nabla$. We assume $|\xi_{i,\ell}|=1$
although it is not strictly necessary.
\medskip

 Homological mirror symmetry associates to the pair $(E,\nabla)$
 a would-be BPS  dyon of a 4d $\cn=2$ gauge theory. 
 The numerical invariants of the
 ODE become the quantum numbers of the would-be dyon. 
 The correspondence is natural in several ways.
In particular:
 \begin{correspondence}\label{kkkkjjjjq}
When the ODE $(E,\nabla)$ reduces to a set of equations of smaller order,
 the would-be BPS dyon becomes marginally
 unstable against decaying into several would-be BPS dyons carrying the quantum numbers of
the smaller order ODEs.
\emph{Stable} BPS dyons correspond to
ODEs with irreducible monodromy $\varrho$.
 \end{correspondence}
 
 In this section we describe as the correspondence arises physically.

\subsection{From a \emph{log}-connection to a 4d $\cn=2$ gauge model}

The first step is to associate a 4d $\cn=2$ gauge theory with gauge group $SU(2)$
to the pair $(E,\nabla)$.
As mentioned before, we restrict to 
non-resonant $\nabla$'s such that the eigenvalues $\xi_{i,\ell}$ of $\varrho_i$ have unit norm.
Since $[\varrho_i]\equiv [\exp(-2\pi i\, A_i)]\subset GL(n,\C)$,
the last condition means that the eigenvalues of the residue matrices $A_i$
are \emph{real} numbers in a transversal of $\R/\Z$. 
We write $E_z$ for the fiber $\simeq \C^n$ of the bundle $E\to\mathbb{P}^1$ over the point $z\in\mathbb{P}^1$.
\medskip

Taken by itself, the holomorphic bundle $E$ is an object in the Abelian category $\mathsf{Coh}\,\mathbb{P}^1$
 of \emph{coherent
sheaves over $\mathbb{P}^1$.} Its Groethendick class $[E]\in K^0(\mathsf{Coh}\,\mathbb{P}^1)\simeq\Z^2$
is determined by two integers: its \emph{rank} $\mathrm{rank}\,E\equiv n$ and its \emph{degree} $\deg E$, cf.\! eq.\eqref{i987weqq0}.

Homological mirror symmetry identifies the derived category $D^b\mathsf{Coh}\,\mathbb{P}^1$ with the BPS category of 
4d $\cn=2$ \emph{pure}  $SU(2)$ SYM \cite{classification,complete}: by this we mean that the BPS particles of $SU(2)$ SYM are
 irreducible families of objects of $D^b\mathsf{Coh}\,\mathbb{P}^1$ which are \emph{stable} with respect to a Bridgeland stability condition \cite{bridgeland}
which depends on the point $u$ in the Coulomb branch. In the QFT language $\mathrm{rank}\,E$ (resp.\! $-\deg E$)
becomes the $SU(2)$ magnetic charge (resp.\! the $SU(2)$ electric charge\footnote{\ Here normalized so that the $W$-boson has electric charge $1$.}).
\medskip

The datum of a logarithmic connection $\nabla$
endows $E$ with additional structure encoded in the conjugacy classes
$c_i\equiv [A_i]\subset \mathfrak{gl}(n,\C)$ of the residue matrices $A_i$ at the marked points $z_i\in S$. 
It is natural to see $E$ as an object
in an \emph{enriched} Abelian category which accounts also for this extra structure.
The standard construction is as follows (see e.g.\! \cite{CB3}). For each $i=1,\cdots\!,s$ choose
a polynomial $P_i(w)\in\C[w]$ which is solved by the residue matrix $A_i$, i.e.\! such that 
$P_i(A_i)=0$. The most economical choice will be the \emph{minimal polynomial}
of $A_i$ \cite{gantmaker}, but we allow for general choices. Factorize $P_i(w)$ into linear factors and choose
some order between them
\be\label{juyqw123}
P_i(w)=\prod_{\ell=1}^{p_i}(w-w_{i,\ell}),\qquad p_i\overset{\rm def}{=}\deg P_i(w),
\ee
and set 
\be
E_{i,k}\overset{\rm def}{=}\prod_{\ell=1}^k(A_i-w_{i,\ell})E_{z_i},\qquad
N_{i,k}\overset{\rm def}{=}\dim E_{i,k},\quad i=1,\cdots\!,s,\ k=0,\cdots\!,p_s.\label{filttt}
\ee 
For elegance of the physical interpretation we assume the
$w_{i,\ell}$ to be real, but everything remains true if they are complex numbers. Likewise, while the order of the factors of $P_i(w)$ is actually irrelevant,
at this stage we order them in a non-decreasing order
\be\label{iiii876123}
w_{i,\ell}\leq  w_{i,\ell+1}\quad \text{for all }i\ \text{and }\ell.
\ee
Mathematically it is pretty obvious that the order of the linear factors in \eqref{juyqw123} should be immaterial. 
However changing their order looks a subtle operation from the physical side (a ``BPS wall-crossing''
\cite{W1,W2,W3}\!\!\cite{4d/2d}) so we prefer to start with the order which makes things simpler from that perspective.
\medskip
   
Eq.\eqref{filttt} defines a length $p_i$ filtration of the fiber $E_{z_i}$ at the market point $z_i\in S$.
We see the filtration as a sequence of \emph{injective} maps (inclusions of $\C$-spaces)
\be\label{filtration}
0\equiv E_{i,p_i}\to E_{i,p_i-1}\to\cdots\cdots\to E_{i,1}\to E_{i,0}\equiv E_{z_i}\simeq \C^n.
\ee
By design this sequence encodes the size and multiplicities of the Jordan blocks of $A_i$
for each distinct eigenvalue.
The non-negative integers $N_{i,\ell}$ are non-increasing
\be
N_{i,\ell}\leq N_{i,\ell-1}.
\ee
When $N_{i,\ell}=N_{i,\ell-1}$ the arrow $E_{i,\ell}\to E_{i,\ell-1}$ is an isomorphism and we may
identify the two spaces, shortening the sequence from length $p_i$ to $p_i-1$: the corresponding factor
$(w-w_{i,\ell})$ in  \eqref{juyqw123} is redundant and we can omit it, decreasing the degree of $P_i(w)$ by 1. 
Conversely we may make the filtration longer by inserting isomorphisms in the sequence.
 Note that
\be\label{ttthissw}
\mathrm{tr}\,A_i=\sum_{k=1}^{p_k} w_{i,k}(N_{i,k-1}-N_{i,k})\quad\Rightarrow\quad
-\deg E=
n\sum_{i=1}^s w_{i,1}+\sum_{i=1}^s\sum_{\ell=1}^{p_i-1} N_{i,\ell}(w_{i,\ell+1}-w_{i,\ell}).
\ee
\begin{rem}\label{ttwist}
When the filtration of $E_{z_i}$ has length $p_i=1$ one has $(A_i-w_{i,1})=0$ i.e.\! the residue matrix $A_i$
is proportional to the identity. Twisting $(E,\nabla)$ by a line bundle $\cl$ with connection
\be
\nabla^\cl=d+\sum_i \frac{a_i}{z-z_i}\ \text{with }\sum_i a_i=-\deg\cl\ \text{and }
p_i=1\ \Rightarrow\ a_i=-w_{i,1}
\ee 
we get rid of all marked points with $p_i=1$. Thus we may assume with no loss
that all lengths satisfy $p_i\geq2$. We may also see an unmarked point $z_0\not\in S$
as a marked point of length $p_0=1$ (an apparent singularity).
\end{rem}
\medskip

A holomorphic bundle $E\to\mathbb{P}^1$ equipped with filtrations of lengths $p_1,\cdots\!, p_s$
of the fibers $E_{z_i}$ over the marked points $z_i\in S$ is said to have
a \emph{parabolic structure}\footnote{\ In the language of \cite{seh} this is a mere \emph{quasi-parabolic structure}.} of type $\boldsymbol{p}\equiv(p_1,\cdots\!,p_s)$.
Parabolic bundles are objects in the
Abelian (hereditary) category $\mathsf{Coh}\,\mathbb{X}(\boldsymbol{p})$ 
of coherent sheaves
over a \emph{quasi}-commutative projective curve\footnote{\ The curve $\mathbb{X}(\boldsymbol{p})$ depends on the set $S$ of marked points (modulo projective equivalence). To simplify the notation, we omit the set $S$ from the symbol of the curve, leaving it implicit.} $\mathbb{X}(\boldsymbol{p})$
whose underling space
is still the sphere $\mathbb{P}^1$ but whose sheaves are locally modified at the 
 points $z_i\in S$ \cite{lenzing}. More precisely $\mathbb{X}(\boldsymbol{p})$ is a \emph{weighted projective line 
of type $\boldsymbol{p}$} in the sense of Geigle and Lenzing \cite{weightPL}.\footnote{\ The Representation Theory of weighted projective lines
has been reviewed in the $\cn=2$ QFT context in \cite{think} and in more detail in \cite{shep}.
For further mathematical literature on coherent sheaves over weighted projective lines see e.g.\!
\cite{line1,line2,line3,line4,line5}.}

\begin{correspondence}\label{kkkiii888} To the logarithmic connection $\nabla\colon E\to E\otimes \Omega^1(\log S)$
we associate the 4d $\cn=2$ QFT with Seiberg-Witten ``curve'' the mirror of the quasi-commutative curve
$\mathbb{X}(\boldsymbol{p})$ such that $E\in\mathsf{Coh}\,\mathbb{X}(\boldsymbol{p})$
and the SUSY central charge\footnote{\ Once given the ``curve'', the datum of the central charge is equivalent
to a Seiberg-Witten differential.}
$Z\colon K^0(\mathsf{Coh}\,\mathbb{X}(\boldsymbol{p}))\to \C$
given below.
\end{correspondence}

The resulting 4d $\cn=2$ QFT was constructed in \cite{classification} and studied in \cite{think,shep}.\footnote{\ Their generalizations for arbitrary simply-laced Lie groups were constructed in \cite{infinite1,infinite2}; see also \cite{DelZotto:2015rca}.} It consists of $\cn=2$ $SU(2)$
SYM gauging the diagonal $SU(2)$ flavor symmetry of a $s$-tuple of Argyres-Douglas (AD) systems \cite{AD1,AD2}
in one-to-one correspondence with the marked points $z_i\in S\subset\mathbb{P}^1$.
The AD system associated to the point $z_i$ has Dynkin type $D_{p_i}$ \cite{AD2} where $p_i$
is the length of the filtration \eqref{filtration}. In addition to the $SU(2)$ charge which gets gauged, $\boldsymbol{q}_{i,0}$,
the $i$-th AD system carries other $p_i-1$ integrally quantized charges $\boldsymbol{q}_{i,\ell}$,
$\ell=1,\cdots, p_i-1$. The space of parameters $t_{i,\ell}$ of the $i$-th system (couplings, masses,
Coulomb branch coordinates) has complex
dimension $p_i$ \cite{classification} and the central charge of the $i$-th AD system, taken in isolation, has the form
$Z_i=\sum_{\ell=0}^{p_i-1} t_{i,\ell}\,\boldsymbol{q}_{i,\ell}$.
\medskip

\textbf{Correspondence \ref{kkkiii888}}
 states that the BPS particles of the 4d $\cn=2$ QFT are the objects
of the triangle category $D^b \mathsf{Coh}\,\mathbb{X}(\boldsymbol{p})$ which are \emph{stable}
 for the appropriate SUSY central charge $Z\colon K^0(\mathsf{Coh}\,\mathbb{X}(\boldsymbol{p}))\to \C$.
Its Groethendick group $K^0\equiv K^0(\mathsf{Coh}\,\mathbb{X}(\boldsymbol{p}))$
is identified with the lattice of QFT charges and the antisymmetric part of the Euler form
$K^0\times K^0\to \Z$ with the Dirac electro-magnetic pairing \cite{4d/2d,classification,think}.  
 The dictionary between 
$K^0(\mathsf{Coh}\,\mathbb{X}(\boldsymbol{p}))$ and the physical charges is \cite{think,shep}:
\begin{itemize}
\item $SU(2)$ electric charge $e=-\deg E$;
\item $SU(2)$ magnetic charge $m=\mathrm{rank}\, E$, so $m\equiv n$ the order of the ODE \eqref{ODE1};
\item internal charges of the several AD systems $\boldsymbol{q}_{i,\ell}=N_{i,\ell}\equiv\dim E_{i,\ell}$.
\end{itemize}

\begin{war} The ``topological'' charge $m$ is canonically defined (it is the \emph{defect} in the sense of Ringel \cite{ringel}\!\!\cite{think}). Instead the ``Noether'' charges $e$, $\boldsymbol{q}_{i,\ell}$
are defined only up to linear redefinitions which preserve the Dirac pairing. The above definitions are then conventional.
We stress that they are not -- in general -- the standard ones used in QFT. For instance, when
$p_i=2$ for all $i$, the $\cn=2$ QFT is a Lagrangian model i.e.\! $SU(2)$ SYM coupled to $N_f\equiv s$ fundamentals.
The physical electric charge $e_\text{ph}$ and the Cartan charges $f_i$ of the  $Spin(2N_f)$ flavor symmetry
are\footnote{\ In $D^b\mathsf{Coh}\,\mathbb{X}(\boldsymbol{p})$ the $N_i$ can take negative values.} 
\be
e_\text{ph}=-2\,\deg E-\sum_i N_{i,1},\qquad f_i = N_{i,1},
\ee 
in a normalization where the $W$-boson has $e_\text{ph}=2$ and $f_i=0$.
\end{war}

The resulting $SU(2)$ gauge theory is UV complete iff \cite{classification,think}
\be\label{uuucuv}
\chi(\mathbb{X}(\boldsymbol{p}))\overset{\rm def}{=}2-\sum_{i=1}^s\left(1-\frac{1}{p_i}\right)\geq0.
\ee
Here $\chi(\mathbb{X}(\boldsymbol{p}))\in\mathbb{Q}$ is the \emph{Euler characteristics} of $\mathbb{X}(\boldsymbol{p})$ \cite{weightPL,line1,line2,line3,line4,line5}.
Even when the condition \eqref{uuucuv} is not met, we can still give a meaning to the 4d $\cn=2$ system 
as a subsector of some UV complete model.

\subsection{The BPS quiver: the squid algebra} 
\label{s:squid}

It is well known that the category $D^b\mathsf{Coh}\,\mathbb{X}(\boldsymbol{p})$
is equivalent to the derived category $D^b\mspace{1.5mu}\mathsf{mod}\mspace{2.5mu}\cs(\boldsymbol{p})$ of modules of the \emph{squid algebra}\footnote{\ Again we leave the dependence of $\cs(\boldsymbol{p})$ on the points $\{z_i\}$ implicit.} $\cs(\boldsymbol{p})$ of type $\boldsymbol{p}$ \cite{squid,ringel}, 
namely the path algebra of the squid quiver $\boldsymbol{Q}(\boldsymbol{p})$ over the solid arrows in the figure
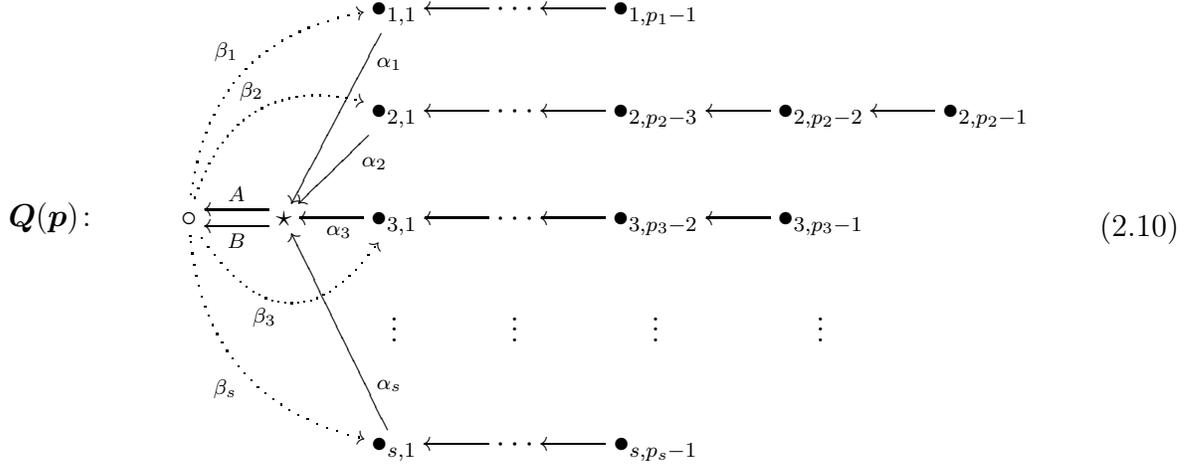
\be\label{biquiver}
\begin{gathered}
\xymatrix{&&& \bullet_{1,1}\ar[ldd]^(0.23){\alpha_1} &\cdots \ar[l] &\bullet_{1,p_1-1}\ar[l]\\
&&& \bullet_{2,1}\ar[ld]^(0.36){\alpha_2} &\cdots \ar[l] &\bullet_{2,p_2-3}\ar[l]&\bullet_{2,p_2-2}\ar[l]&\bullet_{2,p_2-1}\ar[l]\\
\boldsymbol{Q}(\boldsymbol{p})\colon &\circ\ar@{..>}@/^2pc/[uurr]^(0.5){\beta_1}\ar@{..>}@/^2pc/[urr]^(0.5){\beta_2}\ar@{..>}@/_2.7pc/[rr]_(0.4){\beta_3} \ar@{..>}@/_2pc/[ddrr]_(0.5){\beta_s} & \ar@<0.6ex>[l]^B\ar@<-0.6ex>[l]_A\star & \bullet_{3,1}\ar[l]^(0.5){\alpha_3} &\cdots \ar[l] &\bullet_{3,p_3-2}\ar[l]&\bullet_{3,p_3-1}\ar[l]\\
&&& \vdots & \vdots & \vdots &\vdots\\
&&& \bullet_{s,1}\ar[luu]_(0.23){\alpha_s} &\cdots \ar[l] &\bullet_{s,p_s-1}\ar[l]}
\end{gathered}
\ee
bounded by the \emph{squid relations}
\be\label{squid}
(b_i\mspace{2mu} A- a_i\mspace{2mu} B)\alpha_i=0\quad \text{where }\ z_i\equiv (a_i:b_i)\in S\subset \mathbb{P}^1
\ee
which we represent, as customary, by dashed lines $\beta_i$  ($i=1,2\cdots\!,s$) going in the inverse direction.

 The Coxeter form of the BPS quiver (with superpotential),  $(\boldsymbol{Q}(\boldsymbol{p})_\textsc{BPS},\cw)$, of the $\cn=2$ QFT (in the sense of \cite{bpsqui})
is the completion \cite{kellerCYc} of the above squid quiver 
 obtained by making solid the dashed inverse arrows $\beta_i$
 and equipping it with the cubic superpotential
\be\label{uyqertii}
 \cw=\sum_i \mathrm{tr}[\beta_i (b_i\mspace{2mu} A- a_i\mspace{2mu} B)\alpha_i]
\ee
whose derivatives $\partial_{\beta_i}\cw$ reproduce the squid relations \eqref{squid}.
The modules of the squid algebra are precisely the representation of $(\boldsymbol{Q}(\boldsymbol{p})_\textsc{BPS},\cw)$ with vanishing arrows $\beta_i$.
\medskip

A module $Y$ of the squid algebra assigns a vector space\footnote{\ \textbf{Notations:} If $Q$ is any quiver, we write $Q_0$ (resp.\! $Q_1$) for the set of its nodes (resp.\! arrows).\\ The index $\sigma$ (resp.\! $v$)
takes values in the set $\boldsymbol{Q}(\boldsymbol{p})_0\equiv\{\circ,\star,\bullet_{i,\ell}\colon 1\leq i\leq s,\ 1\leq \ell\leq p_i-1\}$ (resp.\! $\boldsymbol{Q}(\boldsymbol{p})_1\equiv\{A,B, \alpha_i, \bullet_{i,\ell}\to \bullet_{i,\ell-1}\colon 1\leq i\leq s,\ 1\leq \ell\leq p_i-1\}$ with the convention that $\bullet_{i,0}\equiv\star$ for all $i$).} $Y_\sigma$ (resp.\! a linear map $Y_v$)
 to each node node $\sigma$
of the quiver (resp.\! to each arrow $v$). The maps are required to
satisfy the squid relations  $(b_i\mspace{2mu} Y_A- a_i\mspace{2mu} Y_B)Y_{\alpha_i}=0$.
 The dimensions of the spaces $Y_\sigma$ give
the conserved charges of the associated BPS state
\be\label{ccchargess}
e=\dim Y_\circ, \qquad m=\dim Y_\star-\dim Y_\circ\equiv n,\qquad \boldsymbol{q}_{i,k}=\dim Y_{\bullet_{i,k}}.
\ee

An indecomposable module $Y$ describes an ODE $(E,\nabla)$ iff the following three conditions hold (see e.g.\! \textbf{Lemma 5.5} of \cite{CB3}):
\begin{itemize}
\item[\it(i)] the maps $a\mspace{2mu} Y_B- b\mspace{2mu} Y_A\colon Y_\star\to Y_\circ$ are surjective for all $(a\colon b)\in\mathbb{P}^1$;
\item[\it(ii)] for $v\neq A,B$ the maps $Y_v$  are injective;
\item[\it(iii)] $\lambda(Y)=0$ where
\be
\lambda(Y)\overset{\rm def}{=}-\dim Y_\circ+(\dim Y_\star-\dim Y_\circ)\sum_{i=1}^s w_{i,1}+\sum_{i=1}^s\sum_{\ell=1}^{p_i-1} (w_{i,\ell+1}-w_{i,\ell})\dim Y_{\bullet_{i,\ell}}.
\ee
\end{itemize}
The function $\lambda\colon K^0(\mathsf{mod}\,\cs(\boldsymbol{p}))\to \R$
 defines a controlled (full) subcategory\footnote{\ \label{iuyqwrtt7}See e.g.\! \textbf{Proposition 1.4} in \cite{keller1}.} $\mathsf{M}\subset\mathsf{mod}\,\cs(\boldsymbol{p})$. 
The objects of $\mathsf{M}$ are the objects $Y\in \mathsf{mod}\,\cs(\boldsymbol{p})$ with
$\lambda(Y)=0$ such that for all subobjects $V$ of $Y$ one has $\lambda(V)\leq 0$. In particular, if $Y\in \mathsf{M}$, all
its indecomposable direct summands $Y_{(a)}$ of $Y$ must also satisfy $\lambda(Y_{(a)})=0$.
\medskip

\noindent\textbf{Notation.} If $Y$ is the squid module associated to the
log-connection $(E,\nabla)$ we write $E_{i,\ell}$ for  the $\C$-space $Y_{\bullet_{i,\ell}}$
and reserve the notation $Y_{i,\ell}$ for the injective arrow $Y_{i,\ell}\colon E_{i,\ell}\to
E_{i,\ell-1}$ with source the node $\bullet_{i,\ell}$ ($\ell\geq1$).

\begin{rem}\label{trivialtwist} Replacing $A_i\to A_i+ a_1\cdot\boldsymbol{1}$ with
$\sum_i a_i=0$ does not change the quiver $\boldsymbol{Q}(\boldsymbol{p})$, the dimension vector $\mathbf{dim}\,Y=(\dim Y_\sigma)$, nor the function $\lambda$. Indeed the replacement
$A_i\to A_i+ a_1\cdot\boldsymbol{1}$ amounts to twisting $E$ by a line bundle, $E\to \cl\otimes E$,
with $\deg\cl=-\sum_ia_i$.
When $\sum_ia_i=0$, $\cl$ is the trivial line bundle $\co$ and the twist acts as the identity. 
\end{rem}

\begin{rem} The representations of the squid algebra form a full subcategory of the representations of the BPS quiver with superpotential $(\boldsymbol{Q}(\boldsymbol{p})_\textsc{BPS},\cw)$.
The stable objects of the second category are the BPS states of the would-be 4d $\cn=2$ QFT,
hence the stable modules of the squid algebra form a subsector of the BPS spectrum.
One may ask whether there is a BPS chamber where this subsector is the full BPS spectrum.
 BPS chambers with this property, dubbed \emph{triangular chambers}, were defined and studied in ref.\!\cite{galois}.
Of course, the question of existence of a triangular chamber for a squid algebra
makes sense only for squid algebras associated to UV-complete 4d QFTs, i.e.\! for squid algebras  satisfying eq.\eqref{uuucuv}.
While a rather special case, these ``physical'' algebras describe some of the most important Fuchsian equations for the
applications: hypergeometric ODEs of order $\leq 4$, the Heun equation, Goursat ODEs of all types but the Appel and the Pochhammer ones, and many others.  
The squid algebras corresponding to UV-complete 4d QFTs are derived equivalent to either affine path algebras (for the asymptotic-free case) or tubular algebras (when the Yang-Mills beta-function vanishes).  
Using the methods of \cite{galois} and the explicit description of the module category
for these two nicer classes of algebras, one can show that there is a \emph{fine-tuned} asymptotic limit in the space of stability functions (which is a zero-coupling limit in a suitable duality frame) where all stable BPS states
are described by modules of the squid algebra \cite{think}. We stress that, typically, this is just an asymptotic limit \emph{not} an actual BPS chamber. By this we mean that one can find a \emph{sequence} of fine-tuned BPS chambers such that the mass of the lightest BPS particle which is \emph{not} a squid module goes to infinity.
Simple examples of this story may be found in ref.\cite{galois}.  Heuristically: differential equations
capture only a (particular) classical limit of the QFT. 
 \end{rem}

\subsection{Affinization of Kac-Moody Lie algebras}
To understand the representation theory of the squid algebra,
we may consider its Brenner matrix \cite{brenner}, namely the \emph{generalized}
symmetric Cartan matrix $B$ with 2's on the main diagonal and non-zero off-diagonal entries
\be
\begin{aligned} 
&B_{\circ,\star}=B_{\star,\circ}=-2, &&B_{\circ,\bullet_{i,1}}=B_{\bullet_{i,1},\circ}=+1,\\
&B_{\star,\bullet_{i,1}}=B_{\bullet_{i,1},\star}=-1,&\quad& B_{\bullet_{i,\ell},\bullet_{i,\ell+1}}=
B_{\bullet_{i,\ell+1},\bullet_{i,\ell}}=-1.
\end{aligned}
\ee 
$B$ is the \emph{Cartan matrix} of the simply-laced Slodowy (a.k.a.\! GIM) Lie algebra
\cite{slodowy} $\mathbb{L}(\boldsymbol{p})$ whose
Dynkin bi-graph is obtained from $\boldsymbol{Q}(\boldsymbol{p})$ by forgetting the orientation of the arrows but keeping the distinction
solid vs.\! dashed. We write $\alpha_\sigma$ for the simple root of $\mathbb{L}(\boldsymbol{p})$ with support at the $\sigma$-th node: $\alpha_\sigma$ stands for the class $[S_\sigma]\in K^0(\mathsf{mod}\,\cs(\boldsymbol{p}))$ of the simple module $S_\sigma$ with
$\dim (S_\sigma)_{\sigma^\prime}=\delta_{\sigma,\sigma^\prime}$.

The Brenner matrix $B$ defines an integral symmetric bilinear pairing $\langle-,-\rangle$
and a quadratic \emph{Tits form} $q(-)$ in $K^0(\mathsf{mod}\,\cs(\boldsymbol{p}))\simeq \oplus_\sigma\Z \mspace{1.5mu}\alpha_\sigma$ 
\be
\langle \xi, \eta\rangle = \xi^t B \eta\in\Z, \qquad q(\eta)\equiv \frac{1}{2}\langle \eta,\eta\rangle,\quad\text{for } \xi,\, \eta\in\bigoplus_{\sigma\in\boldsymbol{Q}(\boldsymbol{p})_0}\Z\mspace{1.5mu} \alpha_\sigma.  
\ee

A \emph{root} of the Slodowy Lie algebra is an element $\eta\in\oplus_\sigma\Z \alpha_v$ with connected support in $\boldsymbol{Q}(\boldsymbol{p})$ and $q(\eta)\leq1$.
If the equality is saturated we say that the root is \emph{real}.
A general result, valid for all \emph{triangular algebras} \cite{delapena}\!\!\cite{galois}, says that
the dimension vector $[Y]\equiv\mathbf{dim}\, Y\equiv (\dim Y_\sigma)_{\sigma\in \boldsymbol{Q}(\boldsymbol{p})_0}$ of an indecomposable module $Y$ is a \emph{positive root},
while the dimension of a rigid indecomposable module is a \emph{real} root. 

In the present
situation we have a more precise result \cite{alg3}.
The Slodowy Lie algebra 
$\mathbb{L}(\boldsymbol{p})$ is isomorphic to the \emph{affinization of the Kac-Moody Lie algebra} $\mathsf{L}(\boldsymbol{p})$
\cite{kac} whose Dynkin graph $\Lambda(\boldsymbol{p})$
 is obtained from the squid quiver $\boldsymbol{Q}(\boldsymbol{p}) $ by omitting the node $\circ$, the arrows inciding on it, and all edge orientations. Explicitly, $\mathbb{L}(\boldsymbol{p})$ is a central extension of the \emph{loop Lie algebra} $\cl\mathsf{L}(\boldsymbol{p})$ of $\mathsf{L}(\boldsymbol{p})$
($\equiv$ the Lie algebra of maps $S^1\to \mathsf{L}(\boldsymbol{p})$) \cite{alg2,alg3}.
\medskip

Let $\Delta(\boldsymbol{p})$ be the set of roots of the Kac-Moody algebra $\mathsf{L}(\boldsymbol{p})$.
The set $\widehat{\Delta}(\boldsymbol{p})$ of roots  of its affinization $\mathbb{L}(\boldsymbol{p})$ is 
\be
\widehat{\Delta}(\boldsymbol{p})=\Big\{\alpha+r\delta\;\Big|\; \alpha\in \Delta(\boldsymbol{p}),\ r\in\Z\Big\}\bigcup\Big\{r\delta\;\Big|\; 0\neq r\in\Z\Big\}.
\ee
The \emph{real} roots of $\mathbb{L}(\boldsymbol{p})$ have the form $\alpha+r\delta$ with $\alpha$ a \emph{real root} of $\mathsf{L}(\boldsymbol{p})$.
The Cartan symmetric bilinear form $\langle-,-\rangle$ on $\Delta(\boldsymbol{p})$ is extended to  
$\widehat{\Delta}(\boldsymbol{p})$ by setting $\langle \delta,-\rangle=0$. The isomorphism $\widehat{\Delta}(\boldsymbol{p})\simeq
\oplus_\sigma\Z\mspace{1.5mu}\alpha_\sigma\equiv K^0(\mathsf{mod}\,\cs(\boldsymbol{p}))$ is
\be
\delta= \alpha_\circ+\alpha_\star,\qquad\Delta(\boldsymbol{p})\equiv \bigoplus_{\sigma\neq\circ} \Z\mspace{2mu} \alpha_\sigma.
\ee
The isotropic root $\delta$ has the quantum numbers of the $SU(2)$ $W$-boson (cf.\! eq.\eqref{ccchargess}).  

A positive root $\eta= \sum_\sigma N_\sigma\mspace{2mu}\alpha_\sigma\in \widehat{\Delta}(\boldsymbol{p})^+$ is called \emph{strict} iff
\be
N_\star \geq N_{\bullet_{i,1}}\geq N_{\bullet_{i,2}}\geq \cdots\geq N_{\bullet_{i,p_i-1}}\geq0\quad\text{for all }i.
\ee
The dimension vectors
of indecomposable parabolic bundles $E$ correspond to strict roots \cite{alg3}
\be\label{0000nvb}
[E]\equiv -\deg E\;\alpha_\circ +(\mathrm{rank}\,E+\deg E)\,\alpha_\star+\sum_{i,\ell} \dim E_{i,j}\;\alpha_{\bullet_{i,\ell}}\in \widehat{\Delta}(\boldsymbol{p})^+.
\ee
The bundle $E$ is rigid iff its root \eqref{0000nvb} is real. The indecomposable parabolic bundle $E$ admits a logarithm connection $\nabla$
whose residue matrices $A_i$ solve the polynomials $P_i(w)\equiv\prod_{\ell=1}^{p_i}(w-w_{i,\ell})$ iff, in addition,
$\lambda([E])=0$. The twist $E\to \cl\otimes E$ by a line bundle makes
\be
[E]\to [E]+ n\deg \cl\;\delta\in K^0(\mathsf{Coh}\,\mathbb{X}(\boldsymbol{p})),
\ee
 so the value of the electric charge matters only mod $n\equiv$ the magnetic charge.
This is just the Witten effect \cite{Witten:1979ey}.

\begin{rem} When $\chi(\mathbb{X}(\boldsymbol{p}))>0$ the Lie algebra $\mathbb{L}(\boldsymbol{p})$ is isomorphic to an ordinary (simply-laced)
affine Lie algebra. If $\chi(\mathbb{X}(\boldsymbol{p}))=0$, $\mathbb{L}(\boldsymbol{p})$ is the affinization of an affine Lie algebra called
an \emph{elliptic}\footnote{\ See \cite{classification} for a discussion in the present context.}
\cite{saito} (or toroidal \cite{alg1}) Lie algebra,
namely an EALA \cite{eala1,eala2} of nullity 2.
\end{rem}

\subsection{The important special case with $c_1(E)=0$}

The most important case, both conceptually and for the applications to physics/geometry,
is when $c_1(E)=0$ (equivalently $\deg E=0$). 
As explained by Simpson \cite{simpson} this is the most interesting case in non-Abelian Hodge theory,
and hence for all related topics such as: variations of Hodge structure, $tt^*$ equations, Seiberg-Witten geometry, etc.
It is also the classical case which leads to Fuchsian equations (the topic of the present paper) as contrasted to more general log-connections.
In particular restricting to $c_1(E)=0$ leads to a much simpler and elegant theory.
The general case $c_1(E)\neq	0$ can also be studied along the same lines,
but the results are considerably more intricate, and they do lead to useful simplifications
of the computations.

In this case
 $(E,\nabla)$ corresponds to a squid module $Y$ with $Y_\circ=Y_A=Y_B=0$
which can be identified
with a representation of the (hereditary) path algebra of the acyclic star-shaped quiver
$Q(\boldsymbol{p})$
\be\label{starquiver}
\begin{gathered}
\xymatrix{&& \bullet_{1,1}\ar[ldd]^(0.23){\alpha_1} &\cdots \ar[l] &\bullet_{1,p_1-1}\ar[l]\\
&& \bullet_{2,1}\ar[ld]^(0.36){\alpha_2} &\cdots \ar[l] &\bullet_{2,p_2-3}\ar[l]&\bullet_{2,p_2-2}\ar[l]&\bullet_{2,p_2-1}\ar[l]\\
Q(\boldsymbol{p}):&\star & \bullet_{3,1}\ar[l]^(0.5){\alpha_3} &\cdots \ar[l] &\bullet_{3,p_3-2}\ar[l]&\bullet_{3,p_3-1}\ar[l]\\
&& \vdots & \vdots & \vdots &\vdots\\
&& \bullet_{s,1}\ar[luu]_(0.23){\alpha_s} &\cdots \ar[l] &\bullet_{s,p_s-1}\ar[l]}
\end{gathered}
\ee
obtained by orienting all edges of the star-shaped 
Kac-Moody Dynkin graph $\Lambda(\boldsymbol{p})$ towards the central node
$\star$.
The squid algebra $\cs(\boldsymbol{p})$ gets replaced by the hereditary algebra $\mathcal{H}(\boldsymbol{p})\equiv\C Q(\boldsymbol{p})$
and the Slodowy Lie algebra $\mathbb{L}(\boldsymbol{p})$ by a more familiar Kac-Moody (KM) algebra
$\mathsf{L}(\boldsymbol{p})$.
Indecomposable ODEs correspond to positive roots of the KM algebra 
and rigid indecomposable ones to real positive roots \cite{kac1,kac2}. 

A module $Y$ of the algebra $\ch(\boldsymbol{p})$ assigns a vector space $Y_\sigma$ to each 
node $\sigma\in Q(\boldsymbol{p})_0$
and a $\dim Y_{t(v)}\times \dim Y_{s(v)}$ complex matrix $Y_v$ to
each arrow $v\in Q(\boldsymbol{p})_1$ with source $s(v)$ and target $t(v)$. Two sets of matrices $\{Y_v\}_{v\in Q(\boldsymbol{p})_1}$
give isomorphic representations iff they are in the same orbit of $G^{\mspace{1mu}\C}\equiv \prod_\sigma GL(\dim Y_\sigma,\C)$.
The subgroup $(\mathsf{End}\,Y)^\times$ acts trivially, so the dimension of the moduli space
of representations with dimension vector $\boldsymbol{d}\equiv (\dim Y_\sigma)_{\sigma\in Q(\boldsymbol{p})_0}$ is
\be
\dim \cm_{\boldsymbol{d}}= \sum_{v\in Q(\boldsymbol{p})_1} d_{t(v)} d_{s(v)}- \sum_{\sigma\in \in Q(\boldsymbol{p})_0} d_\sigma^2+\dim \mathsf{End}\,Y
=\dim\mathsf{End}\, Y- q(\boldsymbol{\dim}\, Y)
\ee
where $q(\boldsymbol{d})= \tfrac{1}{2} \boldsymbol{d}^t C\boldsymbol{d}$ is the Tits form of the Kac-Moody
algebra $\mathsf{L}(\boldsymbol{p})$ with Cartan matrix $C$. When $Y$ is a brick, $\mathsf{End}\, Y\simeq \C$, so
\be\label{uuuu786x}
p(Y)\overset{\rm def}{=} \dim\cm_{\boldsymbol{\dim}\, Y}= 1- q(\boldsymbol{\dim}\, Y)\qquad \textbf{($Y$ a brick)}.
\ee
We may assume the support of $Y$ to be connected. Then $p(Y)$ is non-negative precisely iff the dimension vector of $Y$ is a positive root
of $\mathsf{L}(\boldsymbol{p})$
\be
[Y]\equiv\mathbf{dim}\,Y\in \Delta(\boldsymbol{p})^+,
\ee 
and is
positive iff the root is imaginary. In conclusion:

\begin{correspondence} An \emph{irreducible} logarithmic connection $(E,\nabla)$ with vanishing first 
Chern class, $c_1(E)=0$, corresponds to a
module $Y$ of the hereditary algebra $\ch(\boldsymbol{p})$ with  injective maps
whose
dimension vector 
\be
\alpha\equiv [Y]\overset{\rm def}{=} N_\ast\; \alpha_\star +\sum_{i,\ell\geq1}N_{i,\ell}\,\alpha_{\bullet_{i,\ell}}\equiv\mathrm{rank}\,E\;\alpha_\star+\sum_{i,\ell\geq1}\dim E_{i,\ell}\;\alpha_{\bullet_{i,\ell}}
\ee
 is a 
\emph{strict} positive root $\alpha\in\Delta(\boldsymbol{p})^+$  which satisfies
\be
\lambda(\alpha)\overset{\rm def}{=}N_\star\sum_i w_{i,1}+\sum_i\sum_{\ell=1}^{p_i-1} N_{{i,\ell}}(w_{i,\ell+1}-w_{i,\ell})=0.
\ee
The accessory parameter space $\ca\subset \mathscr{A}$ of the \emph{irreducible} connection
has dimension
\be
\dim_\C \ca= 2\,p(\alpha)\equiv 2-\langle \alpha,\alpha\rangle.
\ee
\end{correspondence}

\subsection{The sign of the Yang-Mills $\beta$-function}

The 4d $\cn=2$ QFT is asymptotically free, superconformal, or IR free iff
the Euler characteristic $\chi(\mathbb{X}(\boldsymbol{p}))$ is, respectively,
positive, zero, or negative: 
indeed  \cite{classification,think}
\be
\beta(g_\text{YM})=-2\, \chi(\mathbb{X}(\boldsymbol{p}))\,g^3_\text{YM}.
\ee
 Equivalently \cite{think,ringel}: 
\begin{itemize}
\item the 4d $\cn=2$ QFT is asymptotically free if and only if
the Kac-Moody algebra $\mathsf{L}(\boldsymbol{p})$ is of \emph{positive-type}
i.e.\! finite-dimensional. In this case the finitely many roots are all real;
\item the 4d $\cn=2$ QFT is a (relevant deformation of a) SCFT if and only
if the Kac-Moody algebra $\mathsf{L}(\boldsymbol{p})$ is of \emph{zero-type}
i.e.\! affine. The imaginary roots have the form $m\delta$ ($m\in\Z$);
\item 4d $\cn=2$ QFT is IR free if and only
if the Kac-Moody algebra $\mathsf{L}(\boldsymbol{p})$ is of \emph{negative-type}. 
\end{itemize}

\begin{rem} This relation between the type of the Kac-Moody algebra and UV-completeness
has a simple physical explanation. By Ringel theorem \cite{ringelthm} the negative simply-laced Kac-Moody Lie algebras are characterized by the fact that their Dynkin graphs have Coxeter elements with spectral radii $>1$.
That is, the number of BPS states \emph{at fixed electric charge} grows exponentially. This contradicts the entropy bound which holds in all UV-complete $d$-dimensional QFT
\be
S(T)\leq \mathrm{const.}\, T^{d-1}\quad \text{(for large temperature $T$).}
\ee
Thus a negative KM algebra should correspond to a non UV-complete $SU(2)$ gauge theory. 
One may be more precise: since $\Lambda(\boldsymbol{p})$ is a star, the characteristic
polynomial $Q(z)$ of its Coxeter elements  is known in closed form. In particular
 \cite{dela}
 \be
 Q(1)= \chi(\mathbb{X}(\boldsymbol{p}))\prod_i p_i,
 \ee
 so 
when $\chi(\mathbb{X}(\boldsymbol{p}))<0$  the Coxeter element has
an odd number of real eigenvalues $>1$.
\end{rem}

\subsection{The stability function}\label{s:stability}

To identify irreducible log-connections $(E,\nabla)$
with \emph{stable} dyons of $SU(2)$ SYM coupled to $s$ AD systems
as claimed in \textbf{Correspondence \ref{kkkkjjjjq}} we need to specify
a stability function $Z$ (SUSY central charge) consistent with
the properties \textit{(i)-(iii)} which characterize the parabolic bundles (cf.\! \S.\,\ref{s:squid}).
We assume that the $w_{i,\ell}$'s are generic real numbers.

\medskip

The simple roots $\alpha_\sigma$'s form a $\Z$-basis of $K^0(\mathsf{mod}\,\cs(\boldsymbol{p}))$,
so to specify a central charge $Z$ it suffices to give the complex numbers $Z_\sigma\equiv Z(\alpha_\sigma)$.
While to find the exact $Z$ for all pairs $(E,\nabla)$ may be hard,
there is a natural \emph{asymptotic}
stability function which works with some limitations to be discussed in a moment.
We take 
\be\label{regime}
Z_\sigma= \Lambda_\sigma\, e^{i\theta+i\, \lambda(\alpha_\sigma)/\Lambda_\sigma} \qquad \text{with }\Lambda_\sigma \gg1\ \text{and }0 <\theta <\pi/2,
\ee
where the $\Lambda_\sigma$'s are positive real parameters which may be chosen at will as long as they are large enough.
Then
\be\label{iiiiiccccc445}
\begin{split}
Z(Y)\equiv&\sum_\sigma\, Z_\sigma\, \dim Y_\sigma\approx\\
&\approx e^{i\theta}\sum_\sigma \Lambda_\sigma\,\dim Y_\sigma +i\mspace{2mu} e^{i\theta}\sum_\sigma \lambda(e_\sigma)\,\dim Y_\sigma \equiv e^{i\theta}\big(\Lambda(Y)+i \lambda(Y)\big).
\end{split}
\ee
 $\Lambda(Y)\gg1$ for all non-zero squid module $Y$. Then
\be\label{iiii887c}
\arg Z(Y)-\theta \approx \frac{\lambda(Y)}{\Lambda(Y)}\ll1.
\ee
A squid module $Y$ is \emph{stable} for the asymptotic stability condition \eqref{iiiiiccccc445} iff for all non-zero proper sub-module $V\subset Y$
one has $\arg Z(V)<\arg Z(Y)$.
If $Y$ is a module with $\lambda(Y)=0$ and we are in the regime \eqref{regime}, $Y$ is \emph{semi}-stable precisely
when  $\lambda(V)\leq 0$ for all submodules $V\subset Y$, i.e.\! when $Y\in \mathsf{M}$
and stable if the inequality is strict. 
We stress that stable modules are \emph{bricks}, i.e.\! $\mathsf{End}(Y)\simeq\C$ (cf.\! footnote \ref{iuyqwrtt7}).

\subparagraph{Limitations.} From the way we got it, the above stability function looks more a first term in an asymptotic expansion at infinite mass  
than an exact expression. To understand its range of validity, we focus on the simpler $c_1(E)=0$ case.
Rescaling the parameters $w_{i,\ell}\to t\,w_{i,\ell}$ with $t>0$ will not change
the asymptotic stability condition since $t$ can be absorbed in $\Lambda(Y)$.
Eq.\eqref{iiiiiccccc445} then looks as a first term in an expansion in powers of $t$,
and
one expects that the actual stability function  contains also higher order terms.
In other words, physical intuition suggests that when the
eigenvalues of the $A_i$'s are small enough \eqref{iiii887c} is reliable. On the other hand stability is an open condition,
and we can safely replace the actual $Z$ by any ``continuous'' deformation of it, where
``continuous'' means that we do not cross walls of marginal instability.
Writing $A_i=t B_i$, we have $\varrho_i = 1+2\pi i t \tilde B_i+O(t^2)$
with $[\tilde B_i]=[B_i]$, and the relation $\varrho_1\cdots\varrho_s=1$ 
yields $\sum_i \tilde B_i+O(t)=0$. As $t\approx 0$ the monodromy representation of $\nabla\equiv d+\sum_i t\, A_i\, d\log(z-z_i)$
agrees asymptotically  with the one for $d+\sum_i t\,\tilde{B}_i\, d\log(z-z_i)$
 which is a Fuchsian ODE with $E\simeq \co^{\,n}$. 
Therefore it is reasonable to expect
that the asymptotic stability function \eqref{iiii887c} is reliable for \emph{all} Fuchsian ODEs
since they are the ones continuously connected to the trivial connection (the topological class of the bundle 
$E$
is obviously invariant under continuous deformations).
In the next section we shall provide more convincing evidence that this is indeed the correct picture.
We conclude this subsection by stating our

\begin{proposal} A Fuchsian system $(\co^{\,n},\nabla)$ corresponds to a would-be
BPS monopole of the $SU(2)$ SYM coupled to AD matter with quantum numbers
$\alpha\equiv [E]\in K^0(\mathsf{mod}\,\C Q(\boldsymbol{p}))\simeq\oplus_{\sigma\neq \circ}\Z\,\alpha_\sigma$ which is semistable for the asymptotic stability function $Z$. 
If the Fuchsian system is irreducible, the BPS monopole is stable.
\end{proposal}

The case $E\not\simeq \co^{\,n}$ should be ``morally'' similar, but we cannot neglect the (unknown) higher order effects.

 \subsection{``Predictions'' from asymptotic freedom}
 
 The purpose of this subsection is to highlight the power of physical intuition
on log-connections
 by showing that deep (true) math theorems can be inferred by heuristic arguments
 so \emph{crude} that they are largely insufficient even for the  \emph{weakest}  
  standard of physicists' rigor.
 
 Suppose the Yang-Mills $\beta$-function of the 4d $\cn=2$ QFT is negative. 
 We are interested in the IR
 physics, and  in this regime the YM electric coupling gets strong. Dually, the magnetic coupling
 becomes small. Hence the gauge coupling of the effective theory living on the world-line
 of a BPS magnetic monopole is asymptotically small in the IR. However this theory has also other couplings
 proportional to the eigenvalues $w_{i,\ell}$ of the $A_i$'s (cf.\! \S.\,\ref{s;L}). 
 Suppose that these couplings also vanish. This automatically forces $e=0$ which (by gauge invariance)
 sets to zero all superpotential interactions.
  Then the effective world-line theory is asymptotically the free theory and therefore $E$ should be the trivial bundle $\co^{\,n}$. In the language of ODEs $\beta<0$ translates in $\mathsf{L}(\boldsymbol{p})$ being a finite-dimensional Lie algebra, and $w_{i,\ell}=0$ in the statement that the local monodromies are
  unipotent. This gives the following special case of Hilbert's $21^\text{st}$ problem:
  
  \begin{pro}[Crawley-Bovey \cite{CB6}]
  Let $(E,\nabla)$ be a logarithmic connection on $\mathbb{P}^1$ whose associated Kac-Moody
  Lie algebra $\mathsf{L}(\boldsymbol{p})$ is finite-dimensional while the local monodromies $\varrho_i$
  are unipotent. Then $E\simeq \co^{\,n}$.
  \end{pro} 
  
The condition that the associated 4d gauge theory is asymptotically-free is essential for the validity of the statement. When the 4d gauge theory has vanishing $\beta$-function, the effective 1d couplings will not go to zero,
and we cannot say anything without making detailed computations. To illustrate the difference between
log-connections associated to 4d $\cn=2$ gauge field theories with $\beta_\textsc{ym}<0$ and $\beta_\textsc{ym}\geq 0$, 
 we present simple examples of irreducible log-connections on $\mathbb{P}^1$
 with $E$ \emph{non-trivial,} whose local monodromies are all unipotent, which means that all other 1d effective couplings associated to the eigenvalues $w_{i,\ell}$ vanish) so that the only non-trivial coupling on the BPS world-line arises from the 4d Yang-Mills interaction. The above \textbf{Proposition} predicts that  the corresponding 
4d gauge theory should have $\beta_\textsc{ym}\geq0$.

As a first example, we consider the rational elliptic surface\footnote{\ We hope no confusion arises between the symbol $E$, which stands for a holomorphic vector bundle over $\mathbb{P}^1$, and the symbol $\ce$
which stands for an elliptic surface, i.e.\! a compact complex manifold of dimension 2 together with
a holomorphic fibration $\pi_1\colon \ce\to\mathbb{P}^1$ whose generic fiber is a smooth
elliptic curve.}
$\ce$ \cite{kunicho1,kunicho2,miranda,mwbook} at a very generic point of its complex moduli. $\ce$ is an elliptic fibration over $\mathbb{P}^1$ with
  12 exceptional fibers all of Kodaira type $I_1$ \cite{fibers}, that is, the Picard-Fuchs equation $\nabla\Pi=0$ for the periods
  $\Pi$ of the elliptic fibers
  is an irreducible log-connection of order 2 with 12 regular singularities whose local monodromies are all unipotent. 
 Since the exceptional fibers are irreducible, the N\'eron and the Weierstrass models coincide \cite{mwbook}, i.e.\! the Weierstrass model
  \be
  y^2=x^3+a x +b
  \ee 
  is smooth. $x$, $y$, $a$, $b$, and the discriminant $-(4a^3+27b^2)$ are (respectively) sections of $\cl^3$, $\cl^2$, $\cl^4$, $\cl^6$ and $\cl^{12}$ for some line bundle $\cl\to\mathbb{P}^1$. Since $\chi(\ce)=12$, $\deg\cl=1$.
The fiber of  $E$ is spanned by $dx/y$ and $x\, dx/y$ so that
\be
E\simeq \cl^{-1}\oplus \cl=\co(-1)\oplus \co(1)\not\simeq \co^{\mspace{2mu}2},
\ee
Our crude argument then says that the $\beta$-function of the corresponding 4d $\cn=2$
QFT should be non-negative. In facts this log-connection is associated to the non-UV-complete $\cn=2$ $SU(2)$ SQCD with $N_f=12$ ($\beta_\textsc{ym}>0$). 
This example may be criticized since the corresponding 4d QFT is not well defined. To make a more convincing case for our claim, we look for examples
of log-connections with unipotent local monodromies which are associated to \emph{bona fide}
4d QFTs and have $E$ non-trivial. The prediction from the crude argument is that all such theories should have $\beta_\textsc{ym}=0$. Before going to that let us generalize the previous (``unphysical'')
example:

\begin{corl}\label{c:4} Let $\ce\to\mathbb{P}^1$ be a $\chi(\ce)\neq0$ elliptic surface with 
section whose fibers are all semi-stable.\footnote{\ I.e.\!  of Kodaira type
$I_{m}$ for $m\geq0$.} The number of exceptional fibers is $\geq4$. When the inequality is saturated
$\ce$ is rational and the corresponding 4d $\cn=2$ gauge theory has $\beta_\text{YM}=0$.
\end{corl}
When the inequality is saturated the Picard-Fuchs equation for the fibers' periods
 has KM algebra $\mathsf{L}(2,2,2,2)\equiv \widetilde{D}_4$ of affine type,
and $[Y]$ is its minimal imaginary root $\delta$. The associated 4d QFT is $SU(2)$ SYM with $N_f=4$ which is \emph{superconformal} ($\beta_\text{YM}=0$).
 Let us show that a ``superconformal'' semi-stable elliptic surface
(with section) $\ce\to\mathbb{P}^1$ is automatically rational. Let $\{I_{m_1},I_{m_2}, I_{m_3}, I_{m_4}\}$
be the list of its singular fibers. The Picard number satisfies 
$\rho(\ce)\geq 2+\sum_a(m_a-1)$ since we have at least that many linear independent
divisor classes (namely the zero section, the generic fiber, and the components of the exceptional fibers which do not cross the zero section) \cite{miranda,mwbook}.\footnote{\ For a review of the geometry of elliptic surfaces from a physical standpoint see \cite{flavor}.} The Euler characteristic is $\chi(\ce)=\sum_a m_a$ and $b_1(\ce)=0$. Hence the second Betti number $b_2(\ce)$ is
\be\label{juqwrt671}
b_2(\ce)=\chi(\ce)-2+b_1(\ce)= \sum_a m_a-2\leq \rho(\ce)\leq b_2(\ce)-2\, p_g(\ce)
\ee
and the geometric genus is $p_g(\ce)\equiv\tfrac{1}{12}\chi(\ce)-1=0$.
In the Persson list of fiber configurations for the rational elliptic surface \cite{fibers}
there are six such ``superconformal'' fiber configurations
\be\label{ffffibe}
\{I_9,3I_1\},\quad \{I_8,I_2,2I_1\},\quad \{I_6,I_3,I_2,I_1\},\quad \{2I_5,2I_1\},\quad \{2I_4,2I_2\},\quad \{4I_3\}.
\ee

More generally eq.\eqref{juqwrt671} yields $2\mspace{1.5mu}p_g(\ce)\leq (s-4)$
for \emph{semi-stable} elliptic surfaces $\ce\to\mathbb{P}^1$ where $s$ is the number of exceptional fibers.
E.g.\! a semi-stable elliptic K3 has at least 6 fibers of type $I_{m_a}$
(cf.\! tables of K3 fibers configurations in \cite{K3}). 

\begin{rem} \textbf{Corollary \ref{c:4}} of asymptotic freedom is a special instance of
the \emph{generalized Arakelov inequality} \cite{ara}, see the \textbf{example} in \S.\,13.4 of \cite{periods}. See also  the main theorem in \cite{beau2}. Eq.\eqref{ffffibe} is the main theorem of \cite{beau3}.
\end{rem}

\subsection{The root set $\Sigma_\lambda$}

Following \cite{CB1,CB2,CB3,CB4} we define
$R_\lambda$ to be the set of positive roots $\alpha\in\Delta(\boldsymbol{p})^+$ of the Kac-Moody
Lie algebra $\mathsf{L}(\boldsymbol{p})$
such that $\lambda(\alpha)=0$. 
We write $\Sigma_\lambda$
for the subset of $\alpha\in R_\lambda$ with the property that 
\be
p(\alpha)> p(\beta_1)+p(\beta_2)+\cdots
\ee
for all decomposition $\alpha=\beta_1+\beta_2+\dots$
 as a sum of two or more elements of $R_\lambda$.
 
\begin{correspondence}\label{coorrr997}
If $(\co^{\,n},\nabla)$ is an irreducible Fuchsian system,
the corresponding BPS monopole has quantum numbers
$[Y] \in\Sigma_\lambda$.
\end{correspondence}

\section{The 4-SUSY SQM $\mathscr{Q}$ and its complexification $\mathscr{Q}^{\,\C}$}
\label{s;L}

We return for a moment to the case of a general logarithmic connection $(E,\nabla)$
which we assume to be ``unitary'' (i.e.\! with Hermitian residues).
Once we have associated to $(E,\nabla)$
a would-be BPS dyonic particle of our $\cn=2$ gauge system
with quantum numbers $(\dim Y_\sigma)_{\sigma\in \boldsymbol{Q}(\boldsymbol{p})_0}$,
we can consider the effective SQM $\mathscr{Q}$ living on its \emph{world-line} $\gamma$.

\subsection{The effective Lagrangian of $\mathscr{Q}$}
Our first problem is to determine the 1d effective Lagrangian $L$ governing the
light degrees of freedom living on $\gamma$. On general grounds \cite{bpsqui} it should be
a quiver 1d gauge theory, invariant under 4-SUSY,
which is based on the 4d BPS quiver $\boldsymbol{Q}(\boldsymbol{p})_\textsc{BPS}$:
each node $\sigma\in \boldsymbol{Q}(\boldsymbol{p})_{\textsc{BPS},0}$
represents a gauge group of rank $\dim Y_\sigma$ and each arrow $v\in \boldsymbol{Q}(\boldsymbol{p})_{\textsc{BPS},1}$ a bifundamental chiral superfield.
 The exact Lagrangian $L$ is
expected to be non-polynomial in the fields and their derivatives, akin to a (supersymmetric) non-Abelian version of
a gauge $+$ matter Born-Infeld Lagrangian. However the common lore says that
for most purposes we may replace the unknown  
non-polynomial Lagrangian $L$ by its lower-dimensional terms which, when written 
in  terms of 4-SUSY superfields, take the schematic form
\be\label{iuytertqwe}
\begin{split}
L_\text{trun}= &\sum_i \int d^4\theta\,\overline{X}_i\, e^{T^a V_a} X_i+\sum_a \lambda_a\mspace{-4mu}\int d^4\theta\, V_a+\\
&\quad+\left(\int d^2\theta\; W(X_i)_\text{cubic} +\text{h.c.}\right)+\int d^2\theta\; \sum_a\mathrm{tr}(W_a^\alpha W_a^\beta)\varepsilon_{\alpha\beta}.
\end{split}
\ee
\emph{A priori} replacing the full non-polynomial Lagrangian $L$ by its 
lower-dimensional truncation $L_\text{trun}$ is appropriate only asymptotically
 in the weak-field
regime, where the chiral scalar fields\footnote{\ Since we are interested in the Higgs phase,  we may assume the vector supermultiplets to be naturally small with no loss.} $X_i$ (and their gradients)
 are small enough to justify neglecting their
higher powers; however it is generally believed that supersymmetry extends the validity of the truncation well beyond the weak-field regime. The standard claim says
that, as long as we are interested only
in SUSY-protected quantities,
we can safely use the truncated Lagrangian $L_\text{trun}$ and still get the \emph{exact} result.
The space of SUSY-preserving vacua is SUSY-protected in a 4-SUSY theory,
so the truncation $L\leadsto L_\text{trun}$  looks ``morally'' justified. However we should be a little more careful. 
The standard claim is based on the fact that the SUSY-protected quantities
are invariant under \emph{continuous deformations} of the Lagrangian
 which preserve SUSY.
We conclude that, \textit{as long as our quantum system $\mathscr{Q}$ is continuously related to a 4-SUSY theory
whose SUSY vacua belong to the parametrically small chiral superfield regime,}
we can use the truncated Lagrangian, \emph{even if the fields $X^i$ are \textbf{not} small,}
and get the \emph{exact} vacuum geometry.
To understand this condition, we rescale all parameters by a common factor $t$
\be
w_{i,\ell}\to t\mspace{2mu} w_{i,\ell}.
\ee
As we shall see momentarily, in the Higgs phase $|X_i|^2=O(t)$ for $t$ small.
 Hence the truncated Lagrangian reproduces the correct Higgs branch
in situations continuously connected to the regime of parametrically small $w_{i,\ell}$.
We already explained in \S.\,\ref{s:stability} that in this regime the bundle $E$ is topologically trivial,
and (of course) the topology of the bundle is invariant under continuous deformations.
We conclude that the truncation $L\leadsto L_\text{trun}$ is legittimate even when the $w_{i,\ell}$ are \textbf{\emph{not}} small  
\emph{provided} the logarithmic
 connection $(E,\nabla)$ is a Fuchsian system.

The present discussion in terms of the physics on the world-line $\gamma$
and the one in \S.\,\ref{s:stability} in terms of the 4d stability function are parallel:
both say that the asymptotic treatment is valid only for situations continuously connected
to the trivial one, i.e.\! for Fuchsian systems. The agreement between the two viewpoints
gives further confidence in the scenario. 
%

\subsection{The truncated effective Lagrangian}
The truncated low-energy effective theory along the BPS particle world-line
 was described in \cite{classification,bpsqui,denef}.
 The SQM of \cite{classification,bpsqui,denef} is a 1d 4-supercharge quiver gauge theory (SQM) over the BPS quiver
$\boldsymbol{Q}(\boldsymbol{p})_\textsc{BPS}$  of the $\cn=2$ theory, that is, over the completion of the 
quiver $\boldsymbol{Q}(\boldsymbol{p})$ in figure \eqref{biquiver}
where now all arrows are drawn solid,
and:
\begin{itemize}
\item the $\sigma$-th node represents a SYM sector with gauge group $U(N_\sigma)$  where $N_\sigma\equiv \dim Y_\sigma$;
\item each arrow $v\in Q_1$ represents a bifundamental chiral superfield $X_v$  in the representation 
\be
\boldsymbol{N}_{t(\sigma)}\times\boldsymbol{\overline{N}}_{s(v)}\quad\text{of}\quad 
U(N_{t(\sigma)})\times U(N_{s(v)}),
\ee
\item the $U(1)$ factor of the gauge group at node $\sigma$ has Fayet-Iliopoulos (FI) coupling $\lambda_\sigma$;
\item the superpotential $\cw(X_v)$ in the one in eq.\eqref{uyqertii}
\be
\cw=\sum_{i=1}^s b_i \,\mathrm{tr}\Big(X_A X_{\alpha_i}X_{\beta_i}\Big)-\sum_{i=1}^s a_i\,\mathrm{tr}\Big(X_B X_{\alpha_i}X_{\beta_i}\Big).
\ee
\end{itemize}

 Half-BPS states, being invariant under 4 supersymmetries, correspond to ground states of $\mathscr{Q}$ which preserves all 
 4 SUSY. The chiral superfield configurations should satisfy the 
 $D$-term conditions
 \be\label{dddterms}
 \sum_{v\colon t(v)=\sigma} X_{v}X^\dagger_{v}- \sum_{v\colon s(v)=\sigma}X_{v}^\dagger X_{v}+\lambda_\sigma +\cdots=0,\qquad
 \text{for all }\sigma\in Q(\boldsymbol{p})_0.
 \ee
 
The overall $U(1)$ group (i.e.\! the diagonal subgroup in the product of all $U(1)$ factors
 at the nodes) is a decoupled \emph{free} 1d vector supermultiplet.
 If the FI coupling of this free subsector
 is non zero, SUSY would be spontaneously broken, and no BPS state would exist. The requirement
 of vanishing overall FI term is most easily obtained by taking the trace of \eqref{dddterms} in the space $\oplus_\sigma Y_\sigma$; one gets\footnote{\ The omitted terms $\cdots$ in eq.\eqref{dddterms} arise from the (unknown) higher order couplings in $L$ which (by definition) do not involve the free $U(1)$ sector.}
 \be
 0=\sum_\sigma \lambda_\sigma\dim Y_\sigma \equiv \lambda(Y)=0,
 \ee  
which is exactly the equation eq.\eqref{i987weqq0} (cf.\! eq.\eqref{ttthissw})
 expressing the topological condition for the existence of the ODE.
When $c_1(E)=0$ the electric charge $-\deg E$ is zero,
 the BPS quiver reduces to $Q(\boldsymbol{p})$
 i.e.\! to the Kac-Moody Dynkin graph with the ingoing orientation, and no superpotential is allowed by gauge invariance.

\subparagraph{Fuchsian ODEs.} We stress again that the above truncated Lagrangian is reliable only in the Fuchsian case, 
 i.e.\! when
 $E\simeq \co^{\,n}$ is trivial. From now on we focus on this case.
 We have automatically $c_1(E)=0$, and the truncated Lagrangian $L_\text{trun}$
 simplifies dramatically. Now eq.\eqref{dddterms}  is the only condition 
 a field configuration should satisfy in order
 to be a \emph{classical} SUSY-preserving vacuum along the Higgs branch (the 
 vector-multiplet scalars being set to zero), and we may safely forget the
 extra terms $\cdots$. We remain with
  \be\label{dddtermsX}
 \sum_{v\colon t(v)=\sigma} X_{v}X^\dagger_{v}- \sum_{v\colon s(v)=\sigma}X_{v}^\dagger X_{v}+\lambda_\sigma =0,\qquad
 \text{for all }\sigma\in Q(\boldsymbol{p})_0.
 \ee

\subsection{The Higgs branch}

 The model $\mathscr{Q}$ has a compact K\"ahler manifold $\mathscr{M}$ of classical vacua. 
More precisely $\mathscr{M}$ in the space of solutions to
\eqref{dddtermsX} modulo the action of the gauge group $G\equiv \prod_\sigma U(\dim Y_\sigma)$).
 At low-energy $\mathscr{Q}$
 effectively reduces to a 1d $\sigma$-model with target space $\mathscr{M}$. 
 When $\mathscr{M}$ is regular enough, the space of \emph{quantum}
 SUSY vacua is  isomorphic \cite{witten} to the
 space of harmonic forms $\mathbb{H}^{\bullet,\bullet}(\mathscr{M})$ of dimension $\sum_{p,q} h^{p,q}(\mathscr{M})$.
$\mathbb{H}^{\bullet,\bullet}(\mathscr{M})$ is the space of 4d BPS states, the Lefshetz $SU(2)$ action on it
 gets identified with spacetime spin rotations \cite{bpsqui} and the action of the Deligne torus\footnote{\ For the definition of the Deligne torus
 see \cite{periods}.} with R-symmetry \cite{michy}. 
 \medskip
 
 For our purposes we are interested directly in the space of classical vacua
 rather than in its cohomology.  
The chiral field configuration
 $\{X_v\}_{v\in Q(\boldsymbol{p})_1}$ breaks the gauge group $G$ down to some subgroup $H$ which contains the overall $U(1)$ which cannot be broken since it is free.
 A configuration ($\equiv$ classical vacuum) with $H=U(1)$ will be called \emph{fully Higgsed.}
 In general the space of fully Higgsed vacua is a submanifold $\cm\subset\mathscr{M}$
 which, \emph{if non-empty,} is open and dense in $\mathscr{M}=\overline{\cm}$.
 Let us compute the complex dimension of $\cm$ in a fully Higgsed phase
  \begin{equation}\label{ddim}
\begin{split}
\dim_\C\cm
=\frac{1}{2}\!
\left(\overbrace{2\sum_{v\in Q(\boldsymbol{p})_1} N_{s(v)} N_{t(v)}}^{\#\ \text{real fields}}- \overbrace{\Big(\sum_{\sigma\in Q(\boldsymbol{p})_0}N_\sigma^2-1\Big)}^{D\ \text{terms eqs.}} -
\overbrace{\Big(\sum_{\sigma\in Q(\boldsymbol{p})_0}N_\sigma^2-1\Big)}^\text{``eaten'' Goldstones}\right)
=1-q(\boldsymbol{N}).
\end{split}
 \end{equation}
 in agreement with eq.\eqref{uuuu786x}.

 \subsection{The complexified quantum system $\mathscr{Q}^{\,\C}$}

 According to \textbf{Correspondence 1} a point in the Higgs branch $\cm$
 with $H=U(1)$ represents an irreducible ``unitary''
  Fuchsian ODE.
 To get the general case where $\Gamma\subset GL(n,\C)$ we have to consider the complexification $\mathscr{Q}^{\,\C}$
 of the SQM model $\mathscr{Q}$ obtained by making the anti-chiral superfields
 to be independent complex superfields and the vectors superfields to be complex.
 The $D$-term equations become
  \be\label{dddterms2}
 \sum_{v\colon t(v)=\sigma} X_{v}X_{v^*}- \sum_{v\colon s(v)=\sigma}X_{v^*} X_{v}+\lambda_\sigma =0,\qquad
 \text{for all }\sigma\in Q(\boldsymbol{p})_0.
 \ee
 where now $X_{v^*}$ are independent $s(v)\times t(v)$ complex matrices,
 and we identify configurations which differ by 
 a $G^{\,\C}\equiv\prod_\sigma GL(n_\sigma, \C)$ gauge rotation. The space of all 
 complex fields $\{X_v,X_{v^*}\}$ now is a holomorphic symplectic manifold with $(2,0)$ form
 \be
 \Omega=i\sum_v \mathrm{tr}(dX_v\wedge dX_{v^*}).
 \ee
 
 The counting of the moduli dimensions in the fully Higgsed  branch with unbroken gauge subgroup $GL(1,\C)\simeq \C^\times$ is the same one as in eq.\eqref{ddim} without the overall factor $1/2$
 \be
 \dim\ca= 2(1-q(\alpha))\equiv 2\,p(\alpha)=2-\langle\alpha,\alpha\rangle.
 \ee
 
 To see that this Lie-theoretic formula agrees with the well-known one in the theory of Fuchsian
 ODEs (for irreducible monodromy), eq.\eqref{887zzaq123},
  consider the $A_{p_i}$ full
 subquiver along  the $i$-th branch of the star-shaped quiver $Q(\boldsymbol{p})$
 \be
 \xymatrix{\star &\bullet_{i,1}\ar[l] &\bullet_{i,2}\ar[l]  &\cdots\cdots\ar[l] &
 \bullet_{i,p_i-1}\ar[l]}
 \ee
 We write $C_i$ for the Cartan matrix of $A_{p_i}$ and 
 $\boldsymbol{N}_i\equiv \alpha|_{A_{p_i}}$ for the restriction of the dimension vector $\alpha\equiv[Y]$ to the
 $A_i$ subquiver.
 Then ($N_\star\equiv n$)
 \be
 n^2+\dim_\C \boldsymbol{\zeta}(A_i)\equiv n^2+\sum_{\ell=0}^{p_i-1}(N_{i,\ell}-N_{i,\ell+1})^2= \boldsymbol{N}^t_i \,C_i\, \boldsymbol{N},
 \ee 
 and 
 \be
 2\,q(\alpha)=\sum_i \boldsymbol{N}^t_i \,C_i\, \boldsymbol{N}-(2s-2)N_\star^2\equiv 
 \sum_i\dim_\C \boldsymbol{\zeta}(A_i)-(s-2)N_\star^2.
 \ee
Note that the chiral superfields $X_v$, seen as linear maps, are injective
in a stable BPS particle.

 \subsection{The deformed preprojective algebra}
 
 A classical SUSY vacuum of the complexified model $\mathscr{Q}^{\,\C}$
 is nothing else than a module $X$ of the \emph{deformed preprojective
 algebra} $\Pi^{\boldsymbol{\lambda}}(\boldsymbol{p})\equiv\Pi^{\boldsymbol{\lambda}}(Q(\boldsymbol{p}))$ \cite{PP1,PP2,PP3,PP4}\footnote{\ \textbf{Warning.} Our convention on the sign of the FI couplings $\lambda_\sigma$ is opposite to the one used in  \cite{PP1,PP2,PP3,PP4}. This is of no consequence since the algebras $\Pi^{\boldsymbol{\lambda}}(\boldsymbol{p})$
 and $\Pi^{-\boldsymbol{\lambda}}(\boldsymbol{p})$ are clearly isomorphic.} of our star-shaped quiver $Q(\boldsymbol{p})$.
A module $X$ of $\Pi^{\boldsymbol{\lambda}}(\boldsymbol{p})$ is a representation of the
\emph{double quiver} $\overline{Q}(\boldsymbol{p})$ of $Q(\boldsymbol{p})$, obtained by replacing each arrow $v\in Q(\boldsymbol{p})_1$ with a pair
of opposite arrows $v,v^*\in\overline{Q}(\boldsymbol{p})_1$
 \begin{equation}
\xymatrix{s(v)\ar@<0.5ex>[rr]^v && t(v)}\quad\leadsto\quad \xymatrix{s(v)\ar@<0.5ex>[rr]^v && \ar@<0.5ex>[ll]^{v^*} t(v).}
 \end{equation} 
A module $X$ of $\Pi^{\boldsymbol{\lambda}}(\boldsymbol{p})$ assigns a vector
space $X_\sigma$ to each node $\sigma\in \overline{Q}(\boldsymbol{p})_0$, and two linear maps (or, concretely, matrices)
 for each arrow $v\in \overline{Q}(\boldsymbol{p})_1$
 \begin{equation}
 X_v\colon X_{s(v)}\to X_{t(v)},\qquad X_{v^*}\colon X_{t(v)}\to X_{s(v)}.
 \end{equation}
 
In the physical set-up, unbroken SUSY requires the $D$-term of each factor gauge group
 to vanish; for the gauge group at the $\sigma$-th node the condition reads  
 \begin{equation}\label{0009zzam}
\sum_{t(v)= \sigma} X_v X_{v^*}-\sum_{s(v)=\sigma}X_{v^*} X_v+ \lambda_{\sigma}\cdot\textbf{Id}_{X_{\sigma}}=0,
 \end{equation}
 
Eq.\eqref{0009zzam} for all $\sigma\in \overline{Q}(\boldsymbol{p})_0$ are the defining relations of the deformed
 preprojective algebra of $\Pi^{\boldsymbol{\lambda}}(\boldsymbol{p})$, and a module $X$ 
 is concretely given by a set of matrices $\{X_v, X_{v^*}\}$ satisfying these constraints. Modules are identified modulo isomorphism,
 i.e.\! modulo the action of $G^{\,\C}\equiv\prod_\sigma GL(\dim X_\sigma,\C)$. 

We conclude that a fully Higgsed vacuum configuration of $\mathscr{Q}^{\mspace{1mu}\C}$
is the same thing as 
 a \emph{simple} module of $\Pi^{\boldsymbol{\lambda}}(\boldsymbol{p})$. Now (cf.\! \textbf{Correspondence \ref{coorrr997}}):

 \begin{thm}[Crawley-Boevey \cite{CB1,PP2}]\label{ttthmm}
 There is a simple representation of $\Pi^{\boldsymbol{\lambda}}(\boldsymbol{p})$ of dimension vector $\alpha$ if and only if $\alpha\in\Sigma_\lambda$. If $\alpha$ is a real root, the simple representation is unique up to isomorphism, and is the only representation of $\Pi^{\boldsymbol{\lambda}}(\boldsymbol{p})$ of dimension vector $\alpha$. If $\alpha$ is an imaginary root there are infinitely many non-isomorphic simple representations.
 \end{thm}

\subsection{From $\mathscr{Q}^{\mspace{2mu}\C}$ to the Fuchsian ODE}

To complete the correspondence between points in the Higgs branch of
$\mathscr{Q}^{\,\C}$ and irreducible Fuchsian ODEs with the numerical invariant
$\alpha\in \Delta(\boldsymbol{p})^+$ we have to show how one
recovers the explicit differential equation from the data $(\mathscr{Q}^{\,\C},a\in\ca)$.
We label the arrows in $Q(\boldsymbol{p})$ by the same symbol $(i,\ell)$ as their source node
and use the notation $E_{i,\ell}$, $X_{i,\ell}$, and $X^*_{i,\ell}$ for 
$X_{\bullet_{i,\ell}}$, $X_{v_{i,\ell}}$, and $X_{v^\ast_{i,\ell}}$ respectively (the notation is consistent since we are identifying
$X_{\bullet_{i,\ell}}$ with the space $E_{i,\ell}$ in the filtration of the fiber $E_{z_i}$ of $E$).
As always in this paper, the points $\{z_1,\cdots\!,z_s\}\subset \mathbb{P}^1$ are fixed once and for all.
\medskip

Let $\{X_{i,\ell},X^*_{i,\ell}\}$ be a configuration 
of the complexified model $\mathscr{Q}^{\,\C}$ which represents the fully Higgsed vacuum $a\in\ca$.
The configuration is a simple module $X$ of the deformed preprojective algebra $\Pi^{\boldsymbol{\lambda}}(\boldsymbol{p})$.
In a simple module all maps $X_{i,\ell}$ are injective and all maps $X^*_{i,\ell}$ are surjective. 
From such a configuration we can reconstruct a Fuchsian ODE by setting
\be\label{jjjuqqqz12z}
A_i=X_{i,1}X_{i,1}^*+w_{i,1} 
\ee
for the residue matrix at $z_i$.
The relation at the node $\star$ gives the Fuchsian relation $\sum_i A_i=0$,
while the relations at the nodes on the $i$-branch (together with the fact that the maps are injective/surjective)
imply that they satisfy $P_i(A_i)=0$, i.e.\! that they belong to the proper conjugacy class $C_i$.
\medskip

Next we show the converse: starting from an irreducible
 ``unitary'' Fuchsian ODE
 with get a point in $\ca$ fixed by the anti-holomorphic involution $\iota$,
 i.e.\! a genuine Higgs vacuum of a physical ($\equiv$ unitary) SQM model.
  
We can rewrite \eqref{jjjuqqqz12z} as the equality of $n\times n$ matrices 
\be
X_{i,1}X_{i,1}^*=A_i-w_{i,1},
\ee
where in the unitary case the matrix $A_i$ is Hermitian.
In our chosen order \eqref{iiii876123} the Hermitian matrix on the \textsc{rhs} has a non-negative spectrum
and rank $N_{i,1}$.
Hence we can find a $n\times N_{i,1}$  matrix $X_{i,1}$ such that $X_{i,1}X_{i,1}^\dagger =A_i-w_{i,1}$
as required for a physical (unitary) vacuum.
Then we have the equality of $N_{i,1}\times N_{i,1}$ matrices
\be\label{iunnnbcv}
X_{i,2}X^*_{i,2}= X_{i,1}^\dagger X_{i,1}-(w_{i,2}-w_{i,1}).
\ee
Again the \textsc{rhs} is a Hermitian matrix of rank $N_{i,2}$ with non-negative spectrum by
eq.\eqref{iiii876123}. We can find a $N_{i,1}\times N_{i,2}$ matrix $X_{i,2}$
such that $X_{i,2}X_{i,2}^\dagger$ is equal to the \textsc{rhs} of \eqref{iunnnbcv}.
Proceeding by induction on $\ell$, we see that we can choose the $G^{\C}$ gauge so that
$X^*_{i,\ell}=X_{i,\ell}^\dagger$, i.e.\! the module is the field configuration of the unitary model $\mathscr{Q}$. 

\begin{rem}\label{rmm4} By a complex gauge rotation we can always set the arrows of
the $X\in\mathsf{mod}\,\Pi^{\boldsymbol{\lambda}}(\boldsymbol{p})$ which describes
a Fuchsian system $(\co^{\,n},\nabla)$ to be the linear maps
which define the parabolic structure on $E$ in eq.\eqref{filtration}, that is,
 \be
\begin{aligned}
& X_{i,\ell}\colon E_{i,\ell}\to E_{i,\ell-1}\ \text{inclusion},\\
& X_{i,\ell}^*=(A_i-w_{i,\ell})\big|_{E_i,\ell-1}\colon E_{i,\ell-1}\to E_{i,\ell}.
\end{aligned}\label{defgauge}
\ee
\end{rem}

\subsection{Reflection functors $\equiv$ Seiberg dualities}\label{s:sseibd}

The fundamental role played by the roots of the Kac-Moody algebra $\mathsf{L}(\boldsymbol{p})$
suggests that the Weyl reflections should be
``symmetries'' of the problem. In facts only a subgroup of the Weyl group
is relevant for our present purposes. 
This is a fundamental aspect of the story which requires some words of  explanation. The Weyl reflections of the Kac-Moody root lattice
correspond physically to Seiberg dualities for the quiver supersymmetric quantum mechanics on the
BPS particle world-line, while from the viewpoint of quivers with superpotentials
they correspond to \emph{reflection functors} between the representation categories of the
relevant quivers with superpotential. These reflection functors are 
explicitly known in full generality \cite{PP4}: however in general they are quite intricate.
Even worse: when a general reflection functor is applied to the representations of a quiver which is describes
a Fuchsian ODE according to our dictionary, it may produce the representations of a quiver with no intepretation in terms of ODEs. 
This representation-theoretical phenomenon reflects the physical fact that a
general Seiberg duality will not only change the ranks of the gauge groups at the nodes
of the quiver but will also add ``mesonic'' degrees of freedom and superpotential
interactions. While the representation-theoretic interpretation of these additional couplings is understood \cite{PP4},
they have no known interpretation in terms of the monodromy representation a Fuchsian ODE.
In this paper we are interested only in dualities which map a 1d SUSY model associated to a 
Fuchsian ODE to a different 1d SUSY model also associated to a Fuchsian ODE. The idea is that
a duality between two physical systems both of which describe Fuchsian ODEs
may be seen as a duality between the differential equations themselves, and
may be used to replace a difficult ODE with a simpler one. Only a subgroup
of the Seiberg dualities has this interpretation of being ``dualities between Fuchsian ODEs'',
and these are the only ones relevant for the present paper. Stated more crudely:
we are only interested in Seiberg dualities which may \emph{simplify} our differential equations,
while the ones which make the equation even nastier or replace it with some
cumbersome still-to-be-understood mathematical entity are of no use for us.
The same logic motivates our restriction to $c_1(E)=0$ (the ``pure monopole'' case):
while there is every reason to believe that the present story has a natural generalization to
$c_1(E)\neq0$ (and Crawley-Boevey did a good job in this direction \cite{CB2,CB3,CB4,CB5,CB6}),
here we focus on combinatoric algorithms which may \emph{explicitly}
simplify (and solve) the ODEs, and when $c_1(E)\neq0$ the would-be
simplifications are not powerful enough to lead to
effective methods for actual computations.  

The Seiberg dualities which preserve the class of SUSY theories with a
simple relation to Fuchsian differential equations are called
\emph{admissible dualities} (or \emph{admissible reflections}).

We say that the reflection at the simple root $\alpha_\sigma$ is \emph{admissible} iff
$\lambda_\sigma\neq0$. We consider only elements of the Weyl group given by
a product of admissible simple reflections.
\medskip

The admissible Weyl reflections $s_\sigma\colon \Delta(\boldsymbol{p})\to \Delta(\boldsymbol{p})$
at a simple root $\alpha_\sigma$ with $\lambda_\sigma\neq0$
 correspond to \emph{reflection functors} \cite{PP2}
 \begin{equation}
\mathfrak{s}_\sigma\colon \mathsf{mod}\,\Pi^{\boldsymbol{\lambda}}(\boldsymbol{p})\to \mathsf{mod}\,\Pi^{\mspace{1.5mu}r_\sigma(\boldsymbol{\lambda})}(\boldsymbol{p}),
 \end{equation}
between the Abelian category of modules of the algebra $\Pi^{\boldsymbol{\lambda}}(\boldsymbol{p})$ defined
by the relations \eqref{0009zzam} with FI couplings $\boldsymbol{\lambda}\equiv (\lambda_\tau)$ and the 
category of modules of the algebra $\Pi^{\mspace{1.5mu}r_\sigma(\boldsymbol{\tau})}(\boldsymbol{p})$
with the reflected FI
couplings $r_\sigma(\boldsymbol{\lambda})\equiv (r_\sigma(\boldsymbol{\lambda})_\tau)$. 
The reflection functors $\mathfrak{s}_\sigma$, and their compositions $\mathfrak{s}_{\sigma_1}\mathfrak{s}_{\sigma_2}\cdots\mathfrak{s}_{\sigma_\ell}$,
are \emph{equivalences} between the ($\C$-linear, Krull-Schmidt, Hom-finite) Abelian categories
$\mathsf{mod}\,\Pi^{\boldsymbol{\lambda}}(\boldsymbol{p})$ and  $\mathsf{mod}\,\Pi^{r_\sigma(\boldsymbol{\lambda})}(\boldsymbol{p})$.  

The reflection functor at a simple root $\alpha_\sigma$ with
$\lambda_\sigma\neq0$ acts on
the dimension vector $\boldsymbol{N}\equiv (N_\tau)$
and the FI coupling $\boldsymbol{\lambda}\equiv (\lambda_\tau)$ as
 \begin{gather}
\boldsymbol{N}\leadsto s_\sigma(\boldsymbol{N})=\boldsymbol{N}-\langle \boldsymbol{N},\alpha_\sigma\rangle\,\alpha_\sigma,\\
\lambda_\tau\leadsto r_\sigma(\lambda_\tau)=\lambda_\tau- \langle \alpha_\tau,\alpha_\sigma\rangle\,\lambda_\sigma.\label{seitransf}
 \end{gather}
 
Next we describe how the reflection functor at $\alpha_\sigma$ acts on the representation $\boldsymbol{X}\in\mathsf{mod}\,\Pi^{\boldsymbol{\lambda}}(\boldsymbol{p})$
given by a vector space $X_\tau$
for each node $\tau$ and a linear map $X_v\colon X_{s(v)}\to X_{t(v)}$
for each arrow in the double quiver $\overline{Q}(\boldsymbol{p})$. 
Since the orientation of $Q(\boldsymbol{p})$ is immaterial, we redefine it
so that $\sigma$ is a \emph{sink} node ($\equiv$ all incident arrows are ingoing).
We define
 \begin{equation}
X_\oplus= \bigoplus_{v\colon t(v)=\sigma} X_{s(v)},\ \qquad \left\{
\begin{aligned}
&\mu_v\colon X_{s(v)}\to X_\oplus &&\text{canonical inclusion}\\
&\pi_v\colon X_{\oplus}\to X_{s(v)} &&\text{canonical projection,}
\end{aligned}\right.
 \end{equation}
and
 \begin{equation}
\mu=\sum_{v\colon t(v)=\sigma} \mu_v\, X_{v^*}, \qquad \pi =-\frac{1}{\lambda_\sigma}\sum_{v\colon t(v)=\sigma} X_v\, \pi_v\qquad\Rightarrow\qquad
\pi\mu=\boldsymbol{1}_{X_\sigma}.
 \end{equation}
The reflection functor at simple root $\alpha_\sigma$ sends the representation $\boldsymbol{X}\in\mathsf{mod}\,\Pi^{\boldsymbol{\lambda}}(\boldsymbol{p})$
 to the representation $\boldsymbol{X^\prime}\in \mathsf{mod}\,\Pi^{r_\sigma(\boldsymbol{\lambda})}(\boldsymbol{p})$ 
 with vector spaces
 \begin{equation}
 X^\prime_\tau= \begin{cases} \mathsf{Im}(\boldsymbol{1}-\mu\pi) & \tau=\sigma\\
 X_\tau & \tau\neq \sigma,
 \end{cases}
 \end{equation}
 and linear maps
 \begin{align}
 X_v^\prime &= \begin{cases} X_v & t(v)\neq \sigma\\
 \lambda_i(1-\mu\pi)\mu_v\colon X^\prime_{s(v)}\to X^\prime_\sigma & \text{otherwise}
 \end{cases}\\
  X_{v^\ast}^{\prime} &= \begin{cases} X_{v^\ast} & t(v)\neq \sigma\\
\pi_v\big|_{X^\prime_\sigma}\colon X^\prime_\sigma\to X^\prime_{s(v)} & \text{otherwise.}
 \end{cases}
 \end{align}
One checks that $\mathfrak{s}_\sigma$ maps simples into simples, i.e.\! it induces an isomorphism of fully Higgsed (complexified) branches
 \begin{equation}
\ca\simeq \ca^\prime
 \end{equation}

We consider two kinds of reflections:
\begin{itemize}
\item reflections at a node $\bullet_{i,\ell}$ ($\ell\geq1$) along the $i$-th branch. We see these reflections as
IR dualities which replace one state of the $i$-th AD matter system with another one.
These reflections are ``internal dualities'' of the $i$-th matter system.
\item the reflection at the node $\star$. This is the \emph{central Seiberg duality}.
\end{itemize}


\subsubsection{Internal IR dualities}
\label{s:internal}

Internal IR dualities of the 
$i$-th matter AD system do not change the ODE but only the way we represent
it in terms of BPS dyons and modules of $\Pi^{\boldsymbol{\lambda}}(\boldsymbol{p})$.
A general internal duality is the composition of a chain of \emph{elementary} ones.

We distinguish 3 kinds of elementary internal duality:
\begin{itemize}
\item[\bf(a)] internal Seiberg duality: the reflection at a simple root $\alpha_{\bullet_{i,\ell}}$
with $\lambda_{\bullet_{i,\ell}}\neq0$ ($\ell\geq1$) along the $i$-th branch of $Q(\boldsymbol{p})$.
See \S.\,\ref{s:uu77776c|} for more properties;
\item[\bf(b1)]adding or deleting a string of zeros at the end of the $i$-th branch
$$
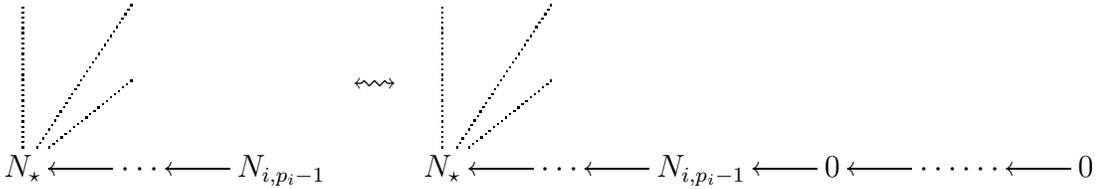
\begin{gathered}
\xymatrix{\ar@{..}[dd] & \ar@{..}[ddl]\\
& \ar@{..}[dl]\\
 N_\star &\ar[l] \cdots &\ar[l] N_{i,p_i-1}}
 \end{gathered}\ \boldsymbol{\leftrightsquigarrow}\ 
\begin{gathered}
\xymatrix{\ar@{..}[dd] & \ar@{..}[ddl]\\
& \ar@{..}[dl]\\
 N_\star &\ar[l] \cdots &\ar[l] N_{i,p_i-1} & \ar[l] 0  &\ar[l]\cdots\cdots &\ar[l] 0}
 \end{gathered}
$$  
\item[\bf(b2)] adding (resp.\! deleting) a new branch to the quiver $Q(\boldsymbol{p})$ with $N_{s+1,\ell}=0$ for all $\ell\geq1$.
This is just the special case of {\bf(b1)} for $p_{s+1}=1$. Up to the twist in \textbf{Remarks \ref{ttwist}, \ref{trivialtwist}} this duality has the effect
of adding a new special point $z_{s+1}$ in the ODE with residue matrix $A_{s+1}=0$.
In other words, we introduce (resp.\! forget) an apparent singularity at $z_{s+1}$ with no effect on the ODE or its solutions.
\end{itemize}

By definition the internal dualities do not change the order $n$ of the ODE. {\bf(a)}, {\bf(b1)}
also preserve the number $s$ of marked points $z_i$ (which may be actual singularities or just apparent ones), while {\bf(b2)} makes $s\to s\pm1$ by adding or deleting an apparent singularity.
\medskip

The main use of internal dualities is to set the AD matter systems in their simpler form,
that is, to give the most economical description of our Fuchsian  ODE. To pursue that goal,
whenever possible, we cancel
all nodes $\tau$ of $Q(\boldsymbol{p})$ with $N_\tau=0$ (\emph{zero-cancellation}).

\subsubsection{Internal Seiberg dualities}
\label{s:uu77776c|}

Recall that we did not require the polynomial $P_i(w)$ to be the minimal one for $A_i$, while we stated
that we can order the linear factors in any way we please (although we chose a particular order for convenience).

\begin{lem} The reflection ($\equiv$ Seiberg duality) at node $\bullet_{i,\ell}$ (with
$\lambda_{i,\ell}\equiv \lambda_{\bullet_{i,\ell}}\neq0$) has the effect
of interchanging the order of the two factors $(w-w_{i,\ell})$ and $(w-w_{i,\ell+1})$ of $P_i(w)$.
\end{lem}

Here we check the statement at the level of the FI couplings. For a complete proof  see \textsc{Appendix}. The new FI couplings are
 \begin{equation}
 \begin{aligned}
 \lambda^\prime_{i,\ell-1}&= \lambda_{i,\ell-1}+\lambda_{i,\ell}= w_{i,\ell+1}-w_{i,\ell-1}\\
 \lambda^\prime_{i,\ell}&= -\lambda_{i,\ell}= w_{i,\ell}-w_{i,\ell+1}\\
 \lambda^\prime_{i,\ell+1}&= \lambda_{i,\ell+1}+\lambda_{i,\ell}= w_{i,\ell+2}-w_{i,\ell}\\
 \lambda^\prime_{i,j}&= \lambda_{i,j}\qquad \text{for }j\neq \ell-1,\ell,\ell+1,
 \end{aligned}
 \end{equation}

\begin{corl}\label{corrr}
Suppose $Y\in\mathsf{mod}\,\Pi^{\boldsymbol{\lambda}}(\boldsymbol{p})$ 
is a simple module of dimension vector
  $\alpha=\sum N_\sigma\alpha_\sigma\in \Sigma_\lambda$ with $N_{i,\ell+1}=N_{i,\ell}$
 \emph{(so that the arrow $Y_{i,\ell+1}\to Y_{i,\ell}$ is an isomorphism).}
By a chain of admissible reflections and zero-cancellations we can replace the $i$-th AD system of type $D_{p_i}$
with the AD system of type $D_{p_i-1}$ where the two nodes with equal dimension ($\equiv$ charge) are
contracted into a single node with a FI coupling which is the sum of the FI couplings of the
contracted nodes. 
\end{corl}
For the proof see \textsc{Appendix}.
\medskip

We conclude that the SQM $\mathscr{Q}$ is independent of the choice of $P_i(w)$ up to internal IR dualities.
We can always reduce to the minimal polynomial $P_i(w)_\text{min}$ of $A_i$.
From now on $p_i$ stands for the degree of the minimal polynomial.
In this case 
 the dimensions along the $i$-th branch of $Q(\boldsymbol{p})$ are positive and strictly decreasing
\be\label{min1}
n>N_{i,1}>N_{i,2}>\cdots > N_{i,p_i-1} >0.
\ee
 
 It is convenient to rearrange the order of the factors in the minimal polynomial so that
 the dimensions $N_{i,\ell}$ take the minimal possible value. This is the condition
 \be\label{min2}
 0\geq \langle \alpha_{\bullet_{i,\ell}},\mathbf{dim}\,Y\rangle \equiv 2N_{i,\ell}-N_{i,\ell-1}-N_{i,\ell+1}\quad \text{for }\ell=1,2,\cdots,p_i-1
 \ee
 (with $N_{i,0}\equiv n$). When eqs.\eqref{min1},\eqref{min2} hold we say that the $i$-th matter system (branch) is
 in the \emph{minimal form} i.e.\! the conjugacy class $c_i$ of $A_i$ is described in the most economic way.
 \medskip
 
 A root $\beta\equiv \sum_\sigma N_\sigma\mspace{1mu} \alpha_\sigma\in \Delta(\boldsymbol{p})^+$
 with all branches set in minimal form will be called \emph{AD-minimal}.
 By construction an AD-minimal $\beta$ is \emph{sincere} (i.e.\! its support is the full quiver)
 and satisfies eqs.\eqref{min1},\eqref{min2} for all $i$.
\medskip

Internal dualities do not modify the special points $z_i$ nor the residue matrices $A_i$ except for forgetting or adding points $z_j$ with $A_j\equiv 0$. Hence internal dualities leave the Fuchsian ODE and its solutions invariant.

\subsubsection{Seiberg duality at the central node $\star$}\label{s:star}

Seiberg duality $\mathfrak{s}_\star$ at $\star$ makes 
\be
N_\star\to N^\prime_\star\equiv \sum_i N_{i,1}-N_\star
\ee
so it changes the order $\equiv N_\star$ of the ODE from $N_\star$ to $N_\star^\prime$. The new root
$s_\star(\beta)$ is not AD-minimal in general, so one has to reduce it to the AD-minimal form by
a suitable sequence of internal Seiberg dualities and zero-cancellations. 
\medskip

At the level of the ODE $\mathfrak{s}_\star$ corresponds to
a transformation called \emph{middle convolution} \cite{katz,convolution,convolution1,convolution2}\!\!\cite{book}.
We review middle convolution in \S.\,\ref{s:solution}.

\subsubsection{Basic ODEs and basic pairs}

The full set of admissible IR dualities acts on the ODEs. In a Seiberg orbit
of irreducible ODEs there is a simplest equation
called the \emph{basic ODE} of the orbit.
The AD-minimal dimension vectors $\beta\equiv\sum_\sigma N_\sigma\mspace{1mu} \alpha_\sigma$
of basic ODEs are characterized by two properties:
\begin{itemize}
\item $\beta\in \Sigma_\lambda$ and sincere;
\item the dimension $\beta$ cannot be further decreased by an \emph{admissible} Seiberg duality.
\end{itemize}
 A pair $(\Lambda,\beta)$ where $\Lambda$ is a star graph and
$\beta$ a sincere positive root is called a \emph{basic pair} iff $(\Lambda,\beta)$ are the
 graph and the AD-minimal dimension vector of
a basic ODE.
\medskip

For \emph{generic} exponents $w_{i,\ell}$ such that $\lambda(\beta)=0$,
 the basic pairs $(\Lambda,\beta)$ are:
\begin{itemize}
\item[(1)] $(A_1,\alpha_1)$ corresponding to the $1^\text{st}$ order ODE.
All indecomposable rigid ODEs are in its duality orbit and $p(\alpha_1)=0$;
\item[(2)] $(\Lambda,\delta)$ with $\Lambda$ an affine Dynkin graph and
$\delta$ the indivisible imaginary root: $p(\delta)=1$;
\item[(3)] $(\Lambda,\beta)$ with $\Lambda$ a Kac-Moody Dynkin graph of negative type
and $\beta\equiv\sum_\sigma N_\sigma\mspace{1mu}\alpha_\sigma$ a sincere element of the fundamental region $\Delta_0$ of the positive roots $\Delta^+$ of
$\Lambda$
\be
\Delta_0\overset{\rm def}{=}\big\{\beta\in \Delta^+,\; \langle \alpha_\sigma,\beta\rangle\leq0\ \ \forall\; \sigma\big\}
\ee  
which is either non-isotropic $q(\beta)<0$ or indivisible $\gcd(N_\sigma)=1$.
\end{itemize}

\section{Action of Seiberg duality on the solutions of ODE}\label{s:solution}

The highly non-trivial action of the central Seiberg duality on the solutions of the Fuchsian equation is given by Katz's middle convolution \cite{katz,convolution,convolution1,convolution2}\!\!\cite{book} as we are going to explain.

\subsection{Review of middle convolution} We quickly review middle convolution.

\subparagraph{Normal form of a Fuchsian ODE.}  Following tradition, we put the Fuchsian system
\eqref{ODE2} in a \emph{normal form} by a trivial twist (cf.\! \textbf{Remark \ref{trivialtwist}}).  First we fix the last regular puncture
$z_s$ at $\infty$ writing the Fuchsian system in the form
\be\label{norm1}
\frac{d}{dz}Y(x)=\sum_{i=1}^{s-1}\frac{A_i}{z-z_i}Y(z),\qquad z_i\in\C,\quad Y(z)=\begin{bmatrix} Y_1(z)\\ \vdots\\ Y_n(z)\end{bmatrix}
\ee
The residue matrix at infinity becomes $A_\infty=-\sum_{i=1}^{s-1}A_i$ by eq.\eqref{Afucchs}.
Then we make the replacements
\be
Y(z)\to Y(z)\prod_{i=1}^{s-1}(z-z_i)^{-w_{i,1}},\qquad A_i\to A_i-w_{i,1},
\ee
so that the normalized $A_i$'s have first eigenvalue zero (\emph{first} in some chosen order).
The first eigenvalue of the new $A_\infty$ is then 
\be\label{norm3}
\lambda_\star\equiv \sum_{i=1}^sw_{i,1}\equiv
\text{the FI coupling at the central node.}
\ee 

\subparagraph{Middle convolution.}
We stress that the Seiberg duality at the central node
$\star$ is admissible if only if the FI coupling $\lambda_\star\neq0$. 
Correspondingly, middle convolution makes sense only in this case,\footnote{\ For the meaning of this condition in the language of ODEs (i.e.\! in terms of properties of the Riemann-Liouville transform) see the original literature on the middle convolution.} and we assume this condition.
\medskip

We define the vector $U(z)$ with $(s-1)n$ components written in terms of $(s-1)$
 blocks $U_j(z)$ of size $n$:
\be\label{kkkjjjjjhg}
U(z)=\begin{bmatrix}U_1(z)\\ \vdots\\ U_{s-1}(z)\end{bmatrix}\qquad
\begin{aligned}\text{where}\quad &U_{i,k}(z)\overset{\rm def}{=}\int_\gamma \frac{Y_k(s)}{s-z_i}(z-s)^{\lambda_\star}\,ds\\
&i=1,2,\cdots\!,s-1,\quad k=1,\cdots\!,n.
\end{aligned}
\ee
For a good choice of the integration contour $\gamma$ \cite{book}
$U(z)$ is a solution to a Fuchsian system of order $(s-1)n$ 
\be\label{pppioza}
\frac{d}{dz}U(x)=\sum_{i=1}^{s-1}\frac{G_i}{z-z_i}U(z)
\ee
where $G_i$ is a $(s-1)n\times (s-1)n$ matrix which we write in terms of $n\times n$ blocks.
The $j,k$ block is 
\be
(G_i)_{jk}= \delta_{ij} \big(A_k+\delta_{ik}\,\lambda_\star\boldsymbol{1}\big),\qquad i,j,k=1,\cdots\!,s-1.
\ee
The ODE \eqref{pppioza} is \emph{highly reducible}. Therefore the entries of the vector $U(x)$ are solutions to {lower-order} \emph{reduced}
equations. Indeed there is a big subspace $V\subset \C^{(s-1)n}$ which is left invariant by all the $G_i$ matrices
\be\label{whatV}
V= \begin{bmatrix}\mathrm{ker}(A_1)\\ \mathrm{ker}(A_2)\\ \vdots\\
\mathrm{ker}(A_{s-1})\end{bmatrix}\oplus \mathrm{ker}(\sum_i G_i)\simeq 
\bigoplus_{i=1}^{s}\mathrm{ker}(A_i) 
\ee
since $\mathrm{ker}(\sum_i G_i)\simeq \mathrm{ker}(A_\infty)$. 
The
dimension of the invariant subspace is
\be
\dim V= sn-\sum_{i=1}^s N_{i,1}\equiv (s-1)n-(\sum_{i=1}^s N_{i,1}-n)\equiv (s-1)N_\star - N_\star^\prime
\ee
We write $P$ for the \emph{constant} $N_\star^\prime\times (s-1)n$ matrix
which yields the canonical projection
\be
P\colon \C^{(s-1)n}\to \C^{(s-1)n}/V\equiv W,\qquad P=\begin{bmatrix}P_1 & P_2 & \cdots & P_{s-1}\end{bmatrix}
\ee
whose blocks $P_j$ are $N^\prime_\star\times n$ matrices. By construction 
\be\label{solution}
Y^\prime(z)\overset{\rm def}{=} P\mspace{1mu}U(z)=\sum_{i=1}^{s-1}\int_\gamma \frac{P_i\mspace{1mu}Y(s)}{s-z_i}\,(z-s)^{\lambda_\star}\,ds
\ee
is a solution to the reduced Fuchsian ODE
of order $n^\prime\equiv N_\star^\prime$ with (possibly apparent) regular singularities 
at the same points $\{z_1,\cdots\!,z_{s-1},\infty\}$ as the original ODE
and new residue matrices $A_i^\prime$
\be\label{whattttA}
A_i^\prime\mspace{1mu} P = P G_i\quad\text{for } i=1,\cdots,s-1,\qquad A_\infty^\prime=-\sum_{j=1}^{s-1}A^\prime_j.
\ee

\subsection{Back to (central) Seiberg duality}
The crucial fact is
\begin{claim} The order $n^\prime$ Fuchsian ODE produced by
middle convolution is the one obtained by Seiberg duality at the central node $\star$.
Therefore, if $Y(z)$ is a solution to the original Fuchsian ODE $(\co^{\,n},\nabla)$,
a solution to the ODE  $(\co^{\,n^\prime},\nabla^\prime)$ obtained from $(\co^{\,n},\nabla)$ by the central Seiberg duality is given explicitly by
eq.\eqref{solution}.
\end{claim}

\begin{proof}
Let us reinterpret in our quiver-theoretic  language the Katz definition of the invariant subspace
$V$ eq.\eqref{whatV}
\be
V\equiv\bigoplus_{i=1}^{s-1}\mathrm{ker}(A_i) \oplus \mathrm{ker}(\sum_i G_i)=
\bigoplus_{i=1}^{s-1} \C^n/E_{i,1}\oplus \C^n/E_{s,1}\simeq \bigoplus_{i=1}^s \C^n/E_{i,1},
\ee
where the last isomorphism shows that the definition is symmetric between all $s$
special points $z_i$. Moreover we see that
\be
W\simeq \left(\bigoplus_{i=1}^s E_{i,1}\right)\bigg/X_\star\simeq \left(\bigoplus_{i=1}^s E_{i,1}\right)\bigg/\mu\pi X_\oplus\equiv X^\prime_\star.
\ee

The residue matrices $A^\prime_i$ of the Seiberg dual ODE, set in their normal form \eqref{norm1}-\eqref{norm3},
are given by \eqref{jjjuqqqz12z}  (here $i=1,2,\cdots\!,s-1$, $w_{i,1}=0$, and no sum over repeated indices)
\be
A^\prime_i\equiv X^\prime_{i,1}X^{\prime\,\ast}_{i,1}=\lambda_\star\,\mu_{i,1}\pi_{i,1}+\sum_{j=1}^s \mu_{j,1}X^*_{j,1}\pi_{i,1}\circeq \lambda_\star\,\mu_{i,1}\pi_{i,1}+\sum_{j=1}^s \mu_{j,1}(A_j- \lambda_\ast\delta_{j,s})\pi_{i,1}
\ee
where we used notations and formulae from \S.\,\ref{s:sseibd} and $\circeq$ means
equality in the defining gauge \eqref{defgauge}. It is easy to see that this formula
coincides with \eqref{whattttA}.
\end{proof}

\begin{correspondence}  If we know the solutions to one irreducible Fuchsian ODE, we can write an integral representation for the
 solutions of all ODEs in its 
Seiberg orbit in terms of reiterated integrals of the form \eqref{solution} (Riemann-Liouville transforms).
\end{correspondence}

All classical integral representations for the solution of special Fuchsian equations, such as
the order-$n$ generalized hypergeometric equations, the Pochhammer equations, the one-variable
reductions of the Appell equations, \emph{etc.,} and also the modern additions to the list \cite{new},
were obtained by methods which are equivalent to a chain of Seiberg dualities.

\subsection{Seiberg dualities and connection formulae}
A theorem by Oshima \cite{AA2,AA2b} states that the connection coefficients
of two ODEs in the same Seiberg duality orbit differ by a factor
which is a ratio of products of Euler's Gamma functions which is explicitly given by the theorem.
In particular, when the ODE is rigid, hence dual to the 1st order equation which has
trivial
connection coefficients, the connection coefficient is given by the 
  Oshima ratio of products of Gamma functions.
 \medskip
 
 More generally, Seiberg duality gives an equivalence between the isomonodromic deformations of dual Fuchsian ODEs \cite{AA3}.
 
 \subsection{Motivic ODEs and Seiberg duality}
 
The local monodromies $\varrho_i$ of a Fuchsian ODE are quasi-unipotent if and only if
 the FI couplings $\lambda_\sigma$ are rational numbers.
 The Seiberg transformation
\eqref{seitransf} maps rational FI couplings into rational ones so
the quasi-unipotency of the local monodromies is preserved by Seiberg duality.
When $\lambda_\sigma\in\mathbb{Q}$ the (closure of the)
 accessory 
parameter space $\mathscr{A}$ is an algebraic variety defined over $\mathbb{Q}$
and a Seiberg duality is a morphisms $\mathscr{A}\to \mathscr{A}^{\,\prime}$
which is also defined over $\mathbb{Q}$. 

Recall that we say that an ODE is motivic if its solutions can be written as 
integrals of differential forms with algebraic coefficients. In particular a motivic ODEs
has quasi-unipotent local monodromies. From eq.\eqref{solution} we see that the
property of being motivic is preserved by Seiberg duality. This implies a special case
of the Deligne-Simposon conjecture: \emph{rigid Fuchsian ODEs are motivic.}

 \section{Examples}
 \label{s:exam}
 
 \subsection{ODEs of order 2}
 
 The minimal form of an order-2 ODE with $s\geq3$ non-apparent regular singularities
 corresponds to a Kac-Moody root of the form
\be\label{oreder2}
\begin{gathered}
 \xymatrix{*++[o][F-]{1}\ar@{-}[dr]& \cdots\cdots & *++[o][F-]{1}\ar@{-}[dl]\\
*++[o][F-]{1}\ar@{-}[r] & *++[o][F-]{2}\\
& *++[o][F-]{1}\ar@{-}[u]}
\end{gathered}
\ee
 where the graph has $s$ peripheral nodes with $N_{i,1}=1$ while $N_\star =2$.
 One has $p(s)=s-3$ and $\dim\ca=2(s-3)$. For $s\geq4$ the root \eqref{oreder2} lays in the fundamental region,
 and the ODE cannot be further simplified.
 \medskip
 
 The BPS quiver $\boldsymbol{Q}(\boldsymbol{p})_\textsc{BPS}$ describes
$\cn=2$ 4d  $SU(2)$ SYM coupled to $N_f\equiv s\geq3$ fundamental hypermultiplets. 
For $N_f>4$ this 4d Lagrangian theory is UV non-complete (but still makes sense as a low-energy effective
theory). The $\cn=2$ model has a flavor symmetry $\mathsf{Spin}(2N_f)$
and the magnetic monopoles\footnote{\ Recall that magnetic monopole is a bit of a misname. The physical electric charge is typically non-zero.} are in the spinorial representation: the state in \eqref{oreder2} has the  highest flavor weight.
Its \emph{conventional} electric charge is zero, but the physical electric charge is $e_\text{ph}\equiv -s$
(in a normalization where the $W$-boson has charge $2$).

\begin{correspondence}
Order-2 irreducible Fuchsian ODEs correspond to BPS dyons of $\cn=2$ $SU(2)$ SQCD with $N_f=s$
 of magnetic charge 2, \emph{physical} electric charge $-N_f$, and flavor spinor highest weight.  
 \end{correspondence}
 
 \subsubsection{Hypergeometric of the $2^\text{nd}$ order: Asymptotic Freedom}
We specialize the above story to $s=3$, i.e.\! $SU(2)$ with $N_f=3$.
Since the 4d QFT is asymptotically free, its Kac-Moody Lie algebra 
 is finite-dimensional, indeed the Lie algebra $D_4\equiv\mathfrak{so}(8)$.
All its roots are real, hence for generic exponents $w_{i,\ell}$
 the ODE is irreducible and rigid.
Therefore this ODE (the hypergeometric equation of the $2^\text{nd}$ order):  \textit{(i)} has a rigid monodromy,
 \textit{(ii)} its solutions have an integral representation, and 
 \textit{(iii)} their connection coefficients are products of $\Gamma$-functions,
 three facts discovered by Riemann in 1851.
 Let us check that they are an easy consequence of Seiberg duality.
The duality at the central node has the effect
$$
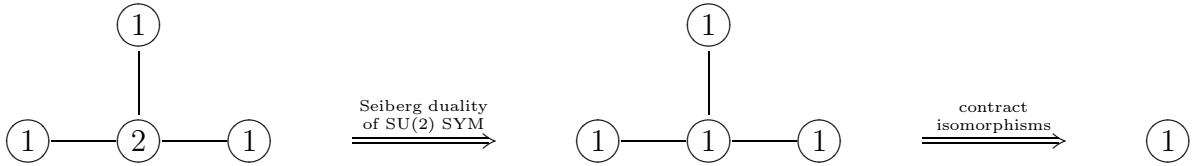
\begin{gathered}
 \xymatrix{& *++[o][F-]{1}\ar@{-}[d] &&&&&& *++[o][F-]{1}\ar@{-}[d]\\
*++[o][F-]{1}\ar@{-}[r] & *++[o][F-]{2}\ar@{-}[r]& *++[o][F-]{1} &\ar@{=>}[rr]^{\text{Seiberg duality}\atop \text{of SU(2) SYM}}
&&& *++[o][F-]{1}\ar@{-}[r] & *++[o][F-]{1}\ar@{-}[r]& *++[o][F-]{1}
&\ar@{=>}[rr]^{\text{contract}\atop \text{isomorphisms}}
&&&*++[o][F-]{1}}
\end{gathered}
$$
and after contracting the isomorphisms (cf.\! \textbf{Corollary \ref{corrr}}) we remain with the simple root $\alpha_\star$, i.e.\! with the $1^\text{st}$ order  equation which is described by the 1d free field theory or, in the 4d language,
 by a free $\cn=2$
 hypermultiplet. The Seiberg duality acts on the FI couplings as
$$
\begin{gathered}
 \xymatrix{& {(c-a-b)}\ar@{-}[d] &&&&&& (c-b)\ar@{-}[d]\\
(-c)\ar@{-}[r] & (a)\ar@{-}[r]& (b-a)&\ar@{=>}[rr]^{\text{Seiberg duality}\atop \text{of SU(2) SYM}}
&&& (a-c)\ar@{-}[r] & (-a)\ar@{-}[r]& (b)}
\end{gathered}
$$
 so that the dual $1^\text{st}$ ODE is
 \be
 \frac{dy}{dz}-\left(\frac{a-c}{z}+\frac{c-b}{z-1}\right)\!y=0\qquad\Rightarrow\qquad y(z)=C\, z^{a-c}\,(z-1)^{c-b},
 \ee
 and from eq.\eqref{solution} one solution to the original ODE is
 \be
Y(z)=\int_\gamma\frac{y(s)}{s-1}\,(s-z)^{\lambda_\star}\,ds= C\!\int_\gamma \frac{s^{a-c}(s-1)^{c-b-1}}{(s-z)^a}\,ds= C\!\int_{\gamma^\prime} \frac{t^{b-1}(1-t)^{c-b-1}\,dt}{(1-zt)^a}
 \ee
 which up to the normalization constant $C$ (and the appropriate choice of contour $\gamma^\prime$)
 is the standard integral representation of the hypergeometric function $F(a,b,c;z)$. 
 \medskip
 
 This result can be easily generalized
 
 \begin{correspondence} An irreducible ODE described by an asymptotically free {\rm(AF)} 4d $\cn=2$ gauge model 
 (i.e.\! $\beta_\text{YM}<0$) is \emph{rigid}, hence Seiberg dual to the $1^\text{st}$ order ODE.  
 \end{correspondence}
The sincere AD-minimal roots for ``asymptotically free'' ODEs
 are\footnote{\ We write the roots as in \textsc{Plance} V-VII of \cite{bourba} which can be used to check the statement.}
 \be\label{finite}
 \begin{smallmatrix}1&2&1\\&1&\end{smallmatrix}\qquad
 \begin{smallmatrix}1&2&3&2&1\\&&1&\end{smallmatrix}\qquad
 \begin{smallmatrix}1&2&4&3&2&1\\&&2&\end{smallmatrix}\qquad
  \begin{smallmatrix}1&3&5&4&3&2&1\\&&2&\end{smallmatrix}\qquad
   \begin{smallmatrix}2&4&6&4&3&2&1\\&&3&\end{smallmatrix}
 \ee
 In the classification of \cite{classification} the corresponding asymptotically-free 4d $\cn=2$ QFTs are called, respectively,
 $SU(2)$ with $N_f=3$, $\widetilde{E}_6$, $\widetilde{E}_7$, $\widetilde{E}_8$
 and $\widetilde{E}_8$.
 
 \begin{rem} Riemann's result is heuristically ``obvious''. When $E\simeq\co^{\,n}$ there is no topological obstruction for
 the world-line theory to flow to the trivial theory, and heuristically the magnetic coupling flows to zero in the IR, so
 it comes as no surprise that its IR effective theory is dual to the free one. 
 \end{rem}

\subsubsection{Heun equation}
 
 For $s=4$ eq.\eqref{oreder2} describes a BPS monopole of $SU(2)$ with $N_f=4$,
 a SCFT with $\beta_\text{YM}=0$. Eq.\eqref{oreder2} is the indivisible imaginary root $\delta$
 of the affine Lie algebra $\widehat{D}_4$, and $p(\delta)=1$.
 The corresponding ODE is  
 the \emph{Heun equation} \cite{heun}. 
 \medskip
 
 $\delta$ belongs to the fundamental region, so the Heun equation is a \emph{basic Fuchsian ODE}
 and its order or number of regular singularities cannot be further reduced.
 In facts the dimension vector $\delta$ is invariant  under the affine Weyl group hence
 all ODEs in its Seiberg orbit are Heun equations.

 However the group of admissible Seiberg dualities (which for generic FI couplings is isomorphic to the
affine Weyl group, hence \emph{infinite}) does act on the FI couplings $\lambda_\sigma$, that is,
on the exponents $\boldsymbol{w}$ of the equation, and also on the accessory parameters $a\in\ca$.
Hence from the solution for one set of parameters $(\boldsymbol{w},a)$ we get
the solution for \emph{infinitely many} other values of the parameters.
For the corresponding statements in a more
classical language and their relations with the theory of the Painlev\'e VI equation, see \cite{jap}.

Then Seiberg duality yields a set of functional equations for the Heun
connection coefficients. Write the action of the Seiberg duality in the form
$(\boldsymbol{w},a)\mapsto (\boldsymbol{w}^\prime,a^\prime)=(\phi(\boldsymbol{w}),f(\boldsymbol{w},a))$
and $C(\boldsymbol{w},a)$ for the connection coefficients. We get a 
functional equation of the schematic form
\be
C(\phi(\boldsymbol{w}),f(\boldsymbol{w},a))=\text{(a known ratio of products of $\Gamma$-functions)}\, C(\boldsymbol{w},a).
\ee
It is plausible that these functional equations are strong enough to determine the 
connection coefficients completely (in principle).
\medskip  

Again these fenomena appear in all ODEs described by the minimal imaginary root $\delta$
of an affine Lie algebra whose Dynkin graph is a star. There are 4 such situations
\be\label{affinebas}
 \begin{smallmatrix}&1\\1&2&1\\&1&\end{smallmatrix}\qquad\quad
  \begin{smallmatrix}1&2&3&2&1\\&&2\\ && 1\end{smallmatrix}\qquad\quad
   \begin{smallmatrix}1&2&3&4&3&2&1\\&&&2&\end{smallmatrix}\qquad\quad
 \begin{smallmatrix}1&2&3&4&5&6&4&2\\&&&&&3&\end{smallmatrix}
\ee
The corresponding $\cn=2$ SCFT are, respectively: $SU(2)$ with $N_f=4$,
$\widetilde{\tilde E}_6$, $\widetilde{\tilde E}_7$, and $\widetilde{\tilde E}_8$
(again in the notation of \cite{classification}).
 
\subsection{Higher order hypergeometrics and Pochhammer equations}

These are two classical $19^\text{th}$ century families of order $n$ rigid indecomposable ODEs which exist for all $n$.
The Thomae order-$n$ hypergeometric ODE \cite{thomae} corresponds to the AD-minimal real root
\be
\begin{gathered}
\xymatrix{&&&& *++[o][F-]{1}\ar@{-}[d]\\
*++[o][F-]{1}\ar@{-}[r]& *++[o][F-]{2}\ar@{-}[r] & *++[o][F-]{3}\ar@{-}[r]&\cdots\ar@{-}[r] &*++[o][F-]{n}\ar@{-}[r]
& \cdots \ar@{-}[r] & *++[o][F-]{3}\ar@{-}[r]& *++[o][F-]{2}\ar@{-}[r]& *++[o][F-]{1}}
\end{gathered}
\ee
For a very detailed description of its monodromy representation see \cite{hyperg}.

A Seiberg duality at the central node, followed by contraction of isomorphisms,
 transforms the hypergeometric ODE of order $n$ into the one of order $(n-1)$
\be
\begin{smallmatrix}&&&&1\\
1&2&\cdots&(n-1)&n&(n-1)&\cdots&2&1
\end{smallmatrix}\ \leadsto\  
\begin{smallmatrix}&&&&1\\
1&2&\cdots&(n-1)&(n-1)&(n-1)&\cdots&2&1
\end{smallmatrix}\ \leadsto\  
\begin{smallmatrix}&&&&1\\
1&2&\cdots&(n-2)&(n-1)&(n-2)&\cdots&2&1
\end{smallmatrix}
\ee

The order $n$ Pochhammer equation \cite{pochh} corresponds to the real root
\be
\begin{gathered}
 \xymatrix{*++[o][F-]{1}\ar@{-}[dr]& \cdots\cdots & *++[o][F-]{1}\ar@{-}[dl]\\
*++[o][F-]{1}\ar@{-}[r] & *++[o][F-]{n}\\
& *++[o][F-]{1}\ar@{-}[u]}
\end{gathered}
\ee
 where the number of peripheral nodes is $s=n+1$. A Seiberg duality at the central node make it into the $1^\text{st}$ order equation.
The Pochhammer equation corresponds to a dyon of $SU(2)$ with $N_f=n+1$ fundamentals with magnetic charge $n$,
 highest flavor weight, and \emph{physical} electric charge $-(n+1)$.

 \subsection{Simpson even/odd series}
 
The infinitely-many sincere \emph{real} positive roots of an affine Lie algebra are easy to write down explicitly: they have the form $\alpha+ r\delta$ for a real root $\alpha$ of the finite-dimensional Lie algebra of the same type.
Since $r$ can be any integer, we get an \emph{infinite sequence} of rigid ODEs of increasing order.
When this strategy is applied to $\widehat{D}_4$ we get the two Simpson series of rigid ODEs
of even resp.\! odd order \cite{sim1,new}.

For order $2n$ we get the series of rigid ODEs
\be\label{ssss1}
\begin{gathered}
 \xymatrix{&*++[o][F-]{\,n\mspace{-5mu}-\mspace{-5mu}1\,}\ar@{-}[d]\\
 *++[o][F-]{n}\ar@{-}[r] &
*++[o][F-]{2n}\ar@{-}[r] & *++[o][F-]{n}\\
& *++[o][F-]{n}\ar@{-}[u]}
\end{gathered}
\ee
for order $2n+1$ we get the series of rigid ODEs
\be\label{ssss2}
\begin{gathered}
 \xymatrix{&*++[o][F-]{n}\ar@{-}[d]\\
 *++[o][F-]{n}\ar@{-}[r] &
*++[o][F-]{\,2n\mspace{-5mu}+\mspace{-4mu}\!1\,}\ar@{-}[r] & *++[o][F-]{n}\\
& *++[o][F-]{n}\ar@{-}[u]}
\end{gathered}
\ee

\begin{correspondence} All Simpson rigid ODEs are described by
a BPS hypermultiplet
of 4d $\cn=2$ $SU(2)$ SYM with $N_f=4$ fundamentals carrying a magnetic charge
equal to the order of the ODE.
\end{correspondence}

Similar constructions give other infinite sequences of rigid ODEs
based on the other 3 affine Lie algebras whose graph is a star:
$\widehat{E}_6$, $\widehat{E}_7$, and $\widehat{E}_8$.
They correspond to magnetically charged hypermultiplets
of the three exceptional complete SCFTs constructed in \cite{classification}:
$\widetilde{\tilde E}_6$,  $\widetilde{\tilde E}_7$, and $\widetilde{\tilde E}_8$.
  
 \subsection{Rigid irreducible ODEs for small $n$}
 
 The irreducible rigid ODEs have been classified (using Kac-Moody techniques)
 for small $n$ in ref.\! \cite{AA2b}. From this reference we take table \ref{tttable}.

 \begin{table}
 \begin{tabular}{c@{\hskip0.5cm}cccccccccccccc}\hline\hline
$n$ & 2 & 3 & 4 & 5 & 6 & 7 & 8 & 9 & 10 & 11 & 12 & 13 & 14 & 15\\
$\#$\text{(rigid ODEs)} & 1 & 2 & 6 & 11 & 28 & 44 & 96 & 157 & 302 & 441 & 857 & 1177 & 2032 & 2841\\\hline\hline
\end{tabular}
\caption{\label{tttable}The number of irreducible rigid Fuchsian ODEs for small order $n$
(from ref.\cite{AA2b})}
 \end{table}

 \subparagraph{Order 3.} At order 3 there are just two irreducible rigid ODE, namely the hypergeometric
 equation, associated to a root of the finite-dimensional Lie algebra $E_6$ (cf.\ eq.\eqref{finite}) 
 and the Pochhammer one which is associated to
 the root $\delta+\alpha_\star$ of $\widehat{D}_4$.

 \subparagraph{Order 4.} At order 4, besides the hypergeometric and the Pochhammer, we have other 4 rigid indecomposable ODEs.
They were already discovered by Goursat in 1886 \cite{goursat} (for a modern account see \cite{account}).
Besides the hypergeometric, associated to the root $\delta-\alpha_s$ of $\widehat{E}_7$ with $\alpha_s$ the simple root in the short branch) and the Pochhammer equation, we have the ODEs in figure
\ref{rigid4} (the roman numeral is
Goursat's symbol for the equation). Except for Goursat type V (whose solutions can be written in terms of Appell functions \cite{NIST}), all other ODEs in the figure describe dyons of a UV-complete (AF or superconformal)
4d QFT.

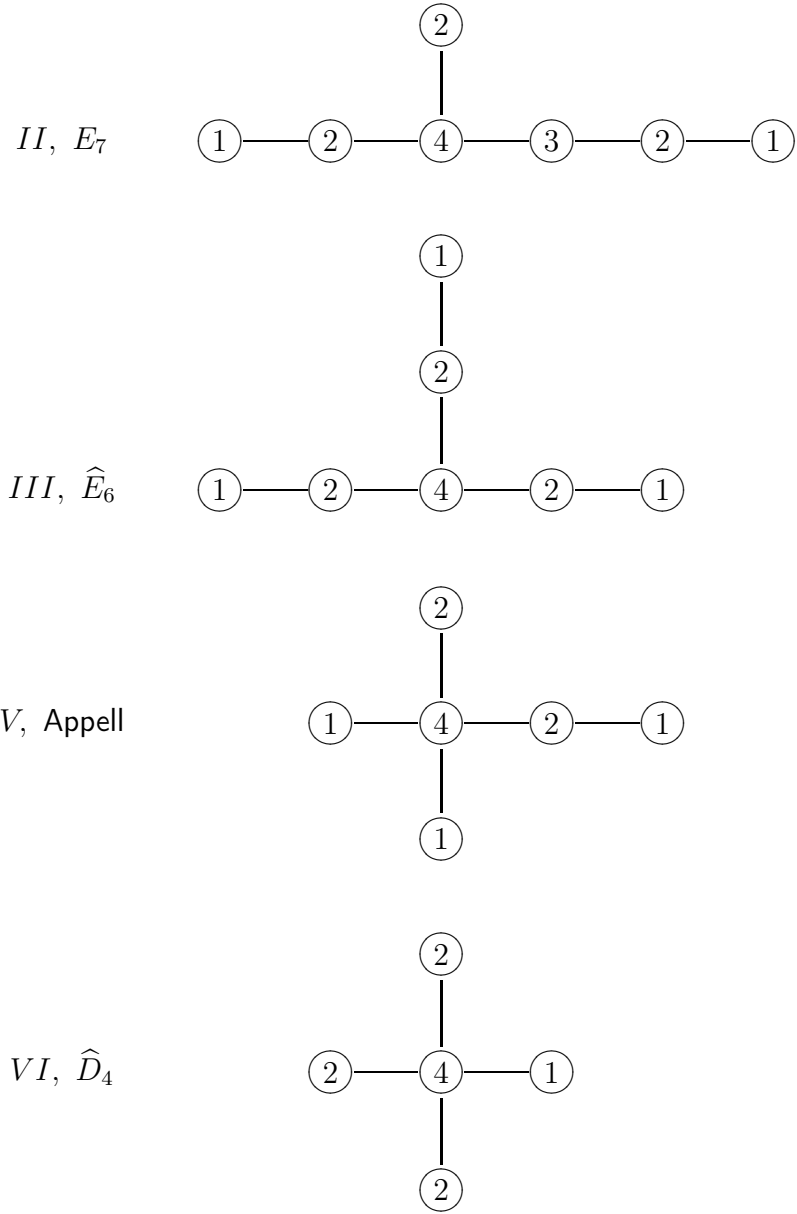
\begin{figure}
$$
\xymatrix{
&&&*++[o][F-]{2}\ar@{-}[d]\\
II,\ {E}_7\hfill &*++[o][F-]{1}\ar@{-}[r]&*++[o][F-]{2}\ar@{-}[r]&*++[o][F-]{4}\ar@{-}[r]&*++[o][F-]{3}\ar@{-}[r]&*++[o][F-]{2}\ar@{-}[r]&*++[o][F-]{1}\\
&&&*++[o][F-]{1}\ar@{-}[d]\\
&&&*++[o][F-]{2}\ar@{-}[d]\\
III,\ \widehat{E}_6\hfill &*++[o][F-]{1}\ar@{-}[r]&*++[o][F-]{2}\ar@{-}[r]&*++[o][F-]{4}\ar@{-}[r]&*++[o][F-]{2}\ar@{-}[r]&*++[o][F-]{1}\\
&&&*++[o][F-]{2}\ar@{-}[d]\\
V,\ \text{\sf Appell}\hfill &&*++[o][F-]{1}\ar@{-}[r]&*++[o][F-]{4}\ar@{-}[r]&*++[o][F-]{2}\ar@{-}[r]&*++[o][F-]{1}\\
&&& *++[o][F-]{1}\ar@{-}[u]\\
&&&*++[o][F-]{2}\ar@{-}[d]\\
VI,\ \widehat{D}_4\hfill &&*++[o][F-]{2}\ar@{-}[r]&*++[o][F-]{4}\ar@{-}[r]&*++[o][F-]{1}\\
&&& *++[o][F-]{2}\ar@{-}[u]\\
}
$$
\caption{\label{rigid4}The Kac-Moody Dynkin graphs
and the corresponding real positive (sincere, AD-minimal) roots
corresponding to four of the six rigid irreducible order-4 Fuchsian equations.
The other two order-4 rigid ODEs not included in the figure 
are the order-4 Thomae hypergeometric (Goursat type I)
and the order-4 Pochhammer ODE (Goursat type VII).}
\end{figure}

\subsection{Seiberg orbits of ODEs}

Classifying Seiberg orbits of irreducible ODEs is the same as classifying basic ODEs, i.e.\!
equations which cannot be further simplified by IR dualities. A basic ODE is specified
by a quadruple $(\Lambda,\boldsymbol{N},\boldsymbol{\lambda},a)$
where
\begin{itemize}
\item $\Lambda$ is a star-shaped Dynkin graph
\item $\boldsymbol{N}=\sum_\sigma N_\sigma \alpha_\sigma$ is an element of $\Sigma_\lambda$
of the corresponding Kac-Moody Lie algebra which cannot be further reduced by Seiberg dualities;
\item $\boldsymbol{\lambda_\sigma}$ is a set of (possibly complex) numbers such that $\sum_\sigma N_\sigma\lambda_\sigma=0$;
\item $a$ is a point in the complexified Higgs branch $\ca$ (the accessory parameter space).
\end{itemize}

The size of a Seiberg orbit depends on the group of admissible Seiberg dualities.
For \emph{generic} exponents the group of the admissible dualities is maximal and coincides with the Weyl group of the Kac-Moody algebra $\mathsf{L}(\boldsymbol{p})$. For generic exponents the basic pairs $(\Lambda,\boldsymbol{N})$ are:
\begin{itemize} 
\item $(A_1,\alpha_1)$ the basic pair in the Seiberg orbit of rigid ODEs ($p(\alpha_1)=0$);
\item $(\Lambda,\beta)$ where $\beta$ is a positive root in the fundamental region, except
that when $\Lambda$ is affine the divisible roots $m \delta$, $m\geq2$ are not basic.
In this case $p(\beta)>0$.
\end{itemize}

On the other extremum, when all $A_i$'s are nilpotents, so that $\lambda_\sigma=0$ at all nodes,
there is no admissible Seiberg duality, and all \emph{irreducible} ODE with nilpotent $A_i$'s
are basic. 
It follows from \textbf{Theorem \ref{ttthmm}} that an irreducible ODE with nilpotent $A_i$ has a pair $(\Lambda,\boldsymbol{N})$
which is basic for generic exponents. Conversely all pairs $(\Lambda,\boldsymbol{N})$ which are basic for generic exponents are also basic when all exponents vanish \emph{with one exception} \cite{CB1}:
\begin{exce}
The Fuchsian ODE with nilpotent local monodromies and pair $(\Lambda,\beta)$
 is \emph{not} irreducible when the corresponding $\Lambda$ is obtained from an affine Dynkin diagram $\varGamma$ by adding a new vertex $w$ connected to an extending vertex of $\varGamma$, the restriction of $\beta$ to $\varGamma$ is a proper multiple of $\delta$, and $N_w=1$.
 \end{exce}
 Thus, while a pair $(\Lambda,\boldsymbol{N})$
which is basic for generic exponents may become reducible when we specialize
the $w_{i,\ell}$, this phenomenon is rather rare and ``typically'' the ODE will remain irreducible even after specialization. Then  the classification of Seiberg orbits for generic exponents suffices for almost all purposes. This classification was already given in \S.\,\ref{s:star}.
\medskip

An important observation is that there is only a \emph{finite number} of (generically) basic pairs $(\Lambda,\beta)$
with a given half-dimension $p(\beta)$ of $\ca$. For instance for $p(\beta)=1$ we have just 4
basic pairs corresponding to the minimal affine imaginary roots in eq.\eqref{affinebas}.

The number of basic pairs $(\Lambda,\beta)$ for small $p(\beta)$ can be read in ref.\cite{AA2b}:
see table \ref{tttable2}.

 \begin{table}
$$
 \begin{tabular}{c@{\hskip0.5cm}ccccccccccc}\hline\hline
$p(\beta)$ & 0 & 1 & 2 & 3 & 4 & 5 & 6 & 7 & 8 & 9 & 10 \\
$\#$\text{(basic pairs)} & 1 & 4 & 13 & 36 & 67 & 90 & 162 & 243 & 305 & 420 & 565\\\hline\hline
\end{tabular}
$$
\caption{\label{tttable2}The number of basic pairs $(\Lambda,\beta)$ with given $p(\beta)$
(taken from ref.\cite{AA2b})}
 \end{table}

 \section*{Acknowledgments}
 I have greatly benefited from long and frutiful discussions with  Michele Del Zotto, Mario Martone, 
 and Robert Moscrop. 
 
 \appendix
 
 \section{Appendix}
 
 We prove the following  
 
 \begin{lem} The AD matter systems defined by different polynomials $P_i(w)$ and/or
 different orders of their linear factors are related by inner IR dualities of the $i$-th AD matter system. 
 \end{lem}

\begin{proof} First of all let us show that reordering the linear factors of $P_i(z)$ amounts to an IR duality.
It is enough to check that switching the order of the $j$-th and $(j+1)$-th linear factors is an admissible
IR duality. Clearly the ranks $N_{i,\ell}$
for $\ell\neq j,j+1$ will not be affected by the switch. $E_{i,j-1}$ is the $\C$-space 
 \begin{equation}
E_{i,j-1}\equiv \prod_{a<j}(A_i -w_{i,a})\C^k,\qquad \dim E_{i,j-1}\equiv N_{i,j-1}.
 \end{equation}
$E_{i,j-1}\subset \C^k$ is an $A_i$-invariant subspace, so the linear endomorphism 
 \begin{equation}
A_i|_{E_{i,j-1}}\colon E_{i,j-1}\to E_{i,j-1}
 \end{equation}
 is well-defined.
We write $N^{\prime}_{i,a}$ for the gauge ranks after the switch of the two linear factors. We have 
 \begin{equation}
\begin{gathered}
N_{i,j}=\dim (A_i-w_{i,j})E_{i,j-1},\qquad N_{i,j+1}=\dim (A_i-w_{i,j+1})(A_i-w_{i,j})E_{i,j-1}\\
N_{i,j}^\prime=\dim (A_i-w_{i,j+1})E_{i,j-1},\qquad N_{i,j+1}^\prime=\dim (A_i-w_{i,j})(A_i-w_{i,j+1})E_{i,j-1}.
\end{gathered}
 \end{equation}
If $w_{i,j}=w_{i,j+1}$ the order switch acts as the identity. Otherwise, let $n(w)$ be
 the number of Jordan blocks of $A_i|_{E_{i,j-1}}$ with eigenvalue $w$. One has
 \begin{align}
 & N_{i,j}= \dim E_{i,j-1}-n(w_{i,j}),&&N_{i,j+1}= \dim E_{i,j-1}- n(w_{i,j})-n(w_{i,j+1})\\
 & N^\prime_{i,j}= \dim E_{i,j-1}-n(w_{i,j+1}),&&N^\prime_{i,j+1}= \dim E_{i,j-1}- n(w_{i,j})-n(w_{i,j+1})
 \end{align} 
 so that only the rank $N_{i,j}$ changes
 \begin{equation}\label{iopr}
 N_{i,j}\leadsto N^\prime_{i,j}= N_{i,j+1}+N_{i,j-1}-N_{i,j}
 \end{equation} 
 since $\dim E_{i,j-1}\equiv N_{i,j-1}$. Eqn.\eqref{iopr} is just the Seiberg duality ($\equiv$ Weyl reflection)
 at the $j$-th node of the $i$-th branch. We have $\lambda_{i,j}=-w_{i,j}+w_{i,j+1}\neq0$,
 and the Seiberg duality at $\alpha_{i,j}$ is an admissible IR duality. The new FI couplings are
 \begin{equation}
 \begin{aligned}
 \lambda^\prime_{i,j-1}&= \lambda_{i,j-1}+\lambda_{i,j}= -w_{i,j-1}+w_{i,j+1}\\
 \lambda^\prime_{i,j}&= -\lambda_{i,j}= -w_{i,j+1}+w_{i,j}\\
 \lambda^\prime_{i,j+1}&= \lambda_{i,j+1}+\lambda_{i,j}= -w_{i,j}+w_{i,j+2}\\
 \lambda^\prime_{i,\ell}&= \lambda_{i,\ell}\qquad \text{for }a\neq j-1,j,j+1,
 \end{aligned}
 \end{equation}
 which agree with the switch $w_{i,j}\leftrightarrow w_{i,j+1}$.
 We conclude that reordering the linear factors corresponds to an inner admissible
 IR duality.
 
 Next we show that replacing the polynomial $P_i(z)$ by a multiple $Q(z)P_i(z)$
 amounts to a admissible IR duality. By the previous result, by an IR
 duality we can transport all linear factors of $Q(z)$ to the left of the linear factors of $P_i(z)$.
 At this point the only effect of $Q(z)$ is to add at the end of the ``matter sector'' linear quiver
 a string of $\deg Q(z)$, zeros since 
 \be
 E_{i,\deg P_i}=E_{i,\deg P_i+1}=\cdots =E_{i,\deg P_i+\deg Q-1}=0.
 \ee 
 We are instructed to cancel the nodes of the quiver which do not belong to the support of $\mathbf{dim}\,E\equiv [E]$.
We remain with the quiver and root of the original polynomial $P_i(z)$ which may be chosen to be the minimal one with no loss. We conclude that the various choices do not
 affect the ``matter system'' modulo admissible IR dualities. 
 \end{proof}

\end{document}